\def\draft{0}
\def\sigconf{0}
\def\big{0}
\def\anon{0}
\def\shownomenclature{0}
\def\masterthesis{0}

\mathchardef\mhyphen="2D

\ifnum\big=1
    \documentclass[final,17pt]{extarticle} \usepackage{fullpage}
\else
    \ifnum\sigconf=1
        \documentclass[sigconf,final]{acmart}
    \else
        \ifnum\draft=1
           \documentclass[11pt,draft,a4paper]{article}
        \else
            \ifnum\masterthesis=1
                \documentclass[11pt,final,a4paper,titlepage]{article}
            \else
                \documentclass[11pt,final,a4paper]{article}
            \fi
        \fi
        \usepackage{lmodern}
    \fi
\fi

\usepackage{setspace}
\usepackage{xcolor}

\title{Uncloneable Decryptors from Quantum Copy-Protection}

\ifnum\sigconf=0\usepackage{authblk}\fi \usepackage{silence} %authbulk must loaded before the silence package!!
\usepackage[modulo]{lineno}
\usepackage{amsmath,amsthm,amssymb,amsfonts,latexsym,color,paralist,url,ifdraft,mathrsfs,thm-restate,comment,bbold,booktabs,dutchcal}
\usepackage{xspace}

\usepackage[utf8]{inputenc} \usepackage[autostyle=false, style=english]{csquotes} \MakeOuterQuote{"}

\usepackage[T1]{fontenc}
\ifnum\sigconf=0
    \ifnum \shownomenclature=1
        \usepackage[refpage,prefix]{nomencl}  
        \makenomenclature
        \usepackage{xpatch}
        \xpatchcmd{\thenomenclature}{%
        \section*{\nomname}% Look for `\section*... etc.
        }{% Replace it by 'nothing'
        \section{\nomname}\label{sec:nomenclature}}{\typeout{Success}}{\typeout{Failure}}
    \fi
    \usepackage[draft=false,pageanchor]{hyperref}
    \usepackage[draft=false,pageanchor]{hyperref}
    
    \usepackage[final]{graphicx}
\fi
\hypersetup{
    colorlinks=true,
    linkcolor=blue,
    filecolor=magenta,
    citecolor={green!75!black},
    urlcolor=cyan,
}

\ifdraft{\usepackage{showlabels}}{} %must be after use package{hyperref}.

\usepackage [ lambda,advantage , operators , sets , adversary , landau , probability , notions , logic ,ff, mm,primitives , events , complexity , asymptotics , keys]{cryptocode}
\theoremstyle{plain}

\ifdraft{\newcommand{\authnote}[3]{{\color{#3} {\bf  #1:} #2}}}{\newcommand{\authnote}[3]{}}

\newcommand{\bra}[1]{\langle #1 \vert}
\newcommand{\ket}[1]{\vert #1 \rangle}

\DeclareMathOperator{\tr}{Tr}

\ifdraft{\linenumbers}{}

\DeclareMathOperator{\E}{\mathbb{E}}

\usepackage{algorithm,algorithmicx,algpseudocode}
\usepackage[normalem]{ulem}

\newcommand{\sR}{\mathcal{R}}
\newcommand{\sO}{\mathcal{O}}
\newcommand{\sC}{\mathcal{C}}
\newcommand{\sD}{\mathcal{D}}
\newcommand{\sA}{\mathcal{A}}
\newcommand{\sB}{\mathcal{B}}

\newcommand{\sP}{\mathcal{P}}
\newcommand{\sF}{\mathcal{F}}

\newcommand{\sE}{\mathcal{E}}

\newcommand{\ds}{\ensuremath{\mathsf{DS}}\xspace}
\newcommand{\qcp}{\ensuremath{\mathsf{QCP}}\xspace}

\newcommand{\se}{\ensuremath{\mathsf{SE}}\xspace}
\newcommand{\bbf}{\ensuremath{\mathsf{BBF}}\xspace}
\newcommand{\ud}{\ensuremath{\mathsf{UD}}\xspace}
\newcommand{\cp}{\ensuremath{\mathsf{QCP}}\xspace}
\renewcommand{\sign}{\ensuremath{\mathsf{Sign}}\xspace}
\newcommand{\ver}{\ensuremath{\mathsf{Ver}}\xspace}
\newcommand{\zo}{\{0,1\}}

\newcommand{\comma}[2]{\begin{aligned}#1  &  ,  &  #2\end{aligned}}

\newcommand{\keygen}{\ensuremath{\mathsf{KeyGen}}\xspace}
\newcommand{\dcs}{\ensuremath{\mathsf{UD}}\xspace}
\newcommand{\dcsob}{\ensuremath{\mathsf{1UD}}\xspace}

\newcommand{\dcsdec}{\ensuremath{\mathsf{UD}^\mathsf{dec}}\xspace}
\newcommand{\dcsobdec}{\ensuremath{\mathsf{1UD}^\mathsf{dec}}\xspace}
\newcommand{\dcsobcca}{\ensuremath{\mathsf{1UD}^\mathsf{cpa}}\xspace}
\newcommand{\dcscpa}{\ensuremath{\mathsf{UD}^\mathsf{cpa}}\xspace}
\newcommand{\dcscca}{\ensuremath{\mathsf{UD}^\mathsf{cpa}}\xspace}
\newcommand{\dcsobccaa}{\ensuremath{\mathsf{1UD}^\mathsf{cca2}}\xspace}
\newcommand{\dcsccaa}{\ensuremath{\mathsf{UD}^\mathsf{cca2}}\xspace}
\newcommand{\udext}{\ensuremath{\mathsf{UD}^\mathsf{ext}}\xspace}

\newcommand{\qenc}{\ensuremath{\mathcal{Enc}}\xspace}
\newcommand{\qdec}{\ensuremath{\mathcal{Dec}}\xspace}
\newcommand{\qsk}{\ensuremath{\mathcal{sk}}\xspace}
\newcommand{\chipgen}{\ensuremath{\mathcal{DecGen}}\xspace}
\newcommand{\decgen}{\ensuremath{\mathcal{DecGen}}\xspace}
\newcommand{\ptect}{\ensuremath{\mathcal{Protect}}\xspace}
\newcommand{\samp}{\ensuremath{\mathsf{Sample}}\xspace}

\newcommand{\qeval}{\ensuremath{\mathcal{Eval}}\xspace}

\newcommand{\sind}{\ensuremath{\mathsf{\mathsf{IND}}}\xspace}
\newcommand{\sindnk}{\ensuremath{\mathsf{\mathsf{IND}}^{n,k}}\xspace}

\newcommand{\sindcpa}{\ensuremath{\mathsf{\mathsf{IND\mhyphen CPA}}}\xspace}

\newcommand{\sindcca}{\ensuremath{\mathsf{\mathsf{IND\mhyphen CCA1}}}\xspace}
\newcommand{\sindcx}{\ensuremath{\mathsf{\mathsf{IND\mhyphen Cx}}}\xspace}

\newcommand{\sindccaa}{\ensuremath{\mathsf{\mathsf{IND\mhyphen CCA2}}}\xspace}

\newcommand{\slorqcp}{\ensuremath{\mathsf{\mathsf{LoR\mhyphen QCP}}}\xspace}

\newcommand{\slorcpa}{\ensuremath{\mathsf{\mathsf{LoR\mhyphen CPA}}}\xspace}

\newcommand{\slor}{\ensuremath{\mathsf{\mathsf{LoR}}}\xspace}

\renewcommand{\prf}{\ensuremath{\mathsf{\mathsf{PRF}}}\xspace}
\newcommand{\sul}{\ensuremath{\mathsf{\mathsf{UL}}}\xspace}
\newcommand{\sful}{\ensuremath{\mathsf{\mathsf{FLIP\mhyphen UL}}}\xspace}

\newcommand{\sud}{\ensuremath{\mathsf{\mathsf{UD}}}\xspace}
\newcommand{\sudnk}{\ensuremath{\mathsf{\mathsf{UD}}^{n,k}}\xspace}

\newcommand{\sudqcpa}{\ensuremath{\mathsf{\mathsf{UD\mhyphen qCPA}}}\xspace}
\newcommand{\sudqcpank}{\ensuremath{\mathsf{\mathsf{UD\mhyphen qCPA}}^{n,k}}\xspace}
\newcommand{\sudqcca}{\ensuremath{\mathsf{\mathsf{UD\mhyphen qCCA1}}}\xspace}
\newcommand{\sudqx}{\ensuremath{\mathsf{\mathsf{UD\mhyphen qX}}}\xspace}
\newcommand{\sudqxnk}{\ensuremath{\mathsf{\mathsf{UD\mhyphen qX}}^{n,k}}\xspace}

\newcommand{\sudqccaa}{\ensuremath{\mathsf{\mathsf{UD\mhyphen qCCA2}}}\xspace}
\newcommand{\sudqccank}{\ensuremath{\mathsf{\mathsf{UD\mhyphen qCCA1}}^{n,k}}\xspace}
\newcommand{\sudqccaank}{\ensuremath{\mathsf{\mathsf{UD\mhyphen qCCA2}}^{n,k}}\xspace}
\newcommand{\sobud}{\ensuremath{\mathsf{\mathsf{UD1}}}\xspace}
\newcommand{\sobudqx}{\ensuremath{\mathsf{\mathsf{UD1\mhyphen qX}}}\xspace}
\newcommand{\sobudqxnk}{\ensuremath{\mathsf{\mathsf{UD1\mhyphen qX}}^{n,k}}\xspace}

\newcommand{\sobudqcpa}{\ensuremath{\mathsf{\mathsf{UD1\mhyphen qCPA}}}\xspace}
\newcommand{\sobudqcpank}{\ensuremath{\mathsf{\mathsf{UD1\mhyphen qCPA}}^{n,k}}\xspace}
\newcommand{\sobudqcca}{\ensuremath{\mathsf{\mathsf{UD1\mhyphen qCCA1}}}\xspace}
\newcommand{\sobudqccaa}{\ensuremath{\mathsf{\mathsf{UD1\mhyphen qCCA2}}}\xspace}
\newcommand{\sobudqccank}{\ensuremath{\mathsf{\mathsf{UD1\mhyphen qCCA1}}^{n,k}}\xspace}
\newcommand{\sobudqccaank}{\ensuremath{\mathsf{\mathsf{UD1\mhyphen qCCA2}}^{n,k}}\xspace}

\newcommand{\sqcp}{\ensuremath{\mathsf{\mathsf{QCP}}}\xspace}
\newcommand{\swqcp}{\ensuremath{\mathsf{\mathsf{WEAK\mhyphen QCP}}}\xspace}
\newcommand{\swqcpnk}{\ensuremath{\mathsf{\mathsf{WEAK\mhyphen QCP}}^{n,k}}\xspace}
\newcommand{\swqcpria}{\ensuremath{\mathsf{\mathsf{WEAK\mhyphen QCP\mhyphen RIA}}}\xspace}
\newcommand{\swqcpriank}{\ensuremath{\mathsf{\mathsf{WEAK\mhyphen QCP\mhyphen RIA}}^{n,k}}\xspace}
\newcommand{\sfdqcp}{\ensuremath{\mathsf{\mathsf{FLIP\mhyphen QCP}}}\xspace}
\newcommand{\sfdqcpnk}{\ensuremath{\mathsf{\mathsf{FLIP\mhyphen QCP}}^{n,k}}\xspace}

\newcommand{\sfdqcpriank}{\ensuremath{\mathsf{\mathsf{FLIP\mhyphen QCP\mhyphen RIA}}^{n,k}}\xspace}
\newcommand{\sidqcp}{\sfdqcp}

\newcommand{\sindx}{\ensuremath{\mathsf{\mathsf{IND\mhyphen X}}}\xspace}

\newcommand{\sqindqx}{\ensuremath{\mathsf{\mathsf{qIND\mhyphen qX}}}\xspace}
\newcommand{\sindqx}{\ensuremath{\mathsf{\mathsf{IND\mhyphen qX}}}\xspace}
\newcommand{\sindqcpa}{\ensuremath{\mathsf{\mathsf{IND\mhyphen qCPA}}}\xspace}
\newcommand{\sqindqcpa}{\ensuremath{\mathsf{\mathsf{qIND\mhyphen qCPA}}}\xspace}
\newcommand{\sqindqcca}{\ensuremath{\mathsf{\mathsf{qIND\mhyphen qCCA1}}}\xspace}
\newcommand{\sqindqccaa}{\ensuremath{\mathsf{\mathsf{qIND\mhyphen qCCA2}}}\xspace}

\newcommand{\sindqcca}{\ensuremath{\mathsf{\mathsf{IND\mhyphen qCCA1}}}\xspace}
\newcommand{\sindqccaa}{\ensuremath{\mathsf{\mathsf{IND\mhyphen qCCA2}}}\xspace}

\newcommand{\sseufcma}{\ensuremath{\mathsf{\mathsf{sEUF\mhyphen CMA}}}\xspace}

\usepackage{tikz}
\usetikzlibrary{arrows,chains,matrix,positioning,scopes}
\makeatletter
\tikzset{join/.code=\tikzset{after node path={%
\ifx\tikzchainprevious\pgfutil@empty\else(\tikzchainprevious)%
edge[every join]#1(\tikzchaincurrent)\fi}}}
\makeatother
\tikzset{>=stealth',every on chain/.append style={join},
         every join/.style={->}}
\tikzstyle{labeled}=[execute at begin node=$\scriptstyle,
   execute at end node=$]
\ifnum\sigconf=0
    \usepackage[capitalise,nameinlink]{cleveref} %Incompatible with showonlyrefs
\fi
\ifnum\sigconf=0
    \newtheorem{theorem}{Theorem}
    \newtheorem{lemma}[theorem]{Lemma}
    \newtheorem{corollary}[theorem]{Corollary}
    \newtheorem{proposition}[theorem]{Proposition}
    \newtheorem{definition}[theorem]{Definition}

\fi 
\newtheorem{claim}{Claim}

\newtheorem*{theorem*}{Theorem}
\newtheorem*{lemma*}{Lemma}
\newtheorem*{corollary*}{Corollary}
\newtheorem*{proposition*}{Proposition}
\newtheorem*{claim*}{Claim}
  
\theoremstyle{definition}

\theoremstyle{remark}
\newtheorem{remark}[theorem]{Remark}
\theoremstyle{plain}

\newtheorem{construction}{Construction}

\usepackage{autonum}%must be loaded after cleveref. Sometimes causes bugs. The alternative is incomaptible with cleveref:
\expandafter\let\expandafter\savedflalignstar\csname flalign*\endcsname
\expandafter\let\expandafter\savedendflalignstar\csname endflalign*\endcsname
\AtBeginDocument{%
  \expandafter\let\csname flalign*\endcsname\savedflalignstar
  \expandafter\let\csname endflalign*\endcsname\savedendflalignstar
}

\newcommand{\onote}[1]{\authnote{Or}{#1}{blue}}
\newcommand{\snote}[1]{\authnote{Shai}{#1}{red}}

\newcommand{\ab}[1]{}  %%

\ifnum\masterthesis=1
    \usepackage{setspace}
    \usepackage{fancyhdr}
\fi
\DeclareMathAlphabet{\mathpzc}{OT1}{pzc}{m}{it}
\usepackage{schemabloc}
\usepackage{ulem}
\usetikzlibrary{shapes,arrows}
\usetikzlibrary{positioning}
\usetikzlibrary{shapes.geometric}
\usepackage{xcolor}
\usepackage{fullpage}

\def\redul#1{\color{red}\underline{{\color{black}#1}}\color{black}}

\begin{document}

\onote{To do list for next version

\begin{enumerate}

        \item Does \sfdqcp imply \slorqcp ?
        
        \item Does it hold for general distributions that \swqcp does not imply \slorqcp? It seems that the splitting attack could be adapted but there are some details in the way.
    
        \item Idea for UD for many bits using \swqcp and iO:
        $\enc_{\sk_\bbf}(m)$ samples $\secpar$ random strings $r_1,\ldots,r_\secpar \gets \zo^{\ell_\bbf}$ and returns $r_1,\ldots,r_\lambda$ along with an obfuscation of
        $$C(x) = 
        \begin{cases} 
            m  & x =\bbf.\eval_{\sk_\bbf}(r_1)\ldots \bbf.\eval_{\sk_\bbf}(r_\secpar) \\
            \bot & \mbox{else}
        \end{cases}$$
        
        The idea is that being iO implies being a eO, that is, if I can distinguish two circuits then I can also find a point where they differ. Hence, if I use the same values of $r_i$ to encrypt two different messages then I can distinguishing the ciphers amounts to finding $\bbf.\eval_{\sk_\bbf}(r_1)\ldots \bbf.\eval_{\sk_\bbf}(r_\secpar)$. This could be used to win the \swqcp game.
        \onote{I'm not sure about serveral things here. First, is eO (which I persume stands for extractable obfuscation) the same as differing input obfuscation? I don't know all the details regarding the implication, and in particular, are there no caveats for that implication? More importantly, can you sketch the security proof? What security are you arguing? CCA-2? That's what you mentioned in your whatsapp message. I don't think this is reasonable, since I think(?) I have an attack for this construction: at most it should be UD-CPA secure. 
        }
        \snote{If my argument above is correct, this construction should be CPA1. CCA2 is clearly unreasonable since there is no guarantee that I can't alter the obfuscated circuit of the challenge cipher a little bit without changing its functionality. What I meant in Whatsapp is that I think this construction gives CCA1 (or at least CPA) from weak QCP, and then we can transform it to CCA2 usind sEUF-CMA secure signatures (and wondered whether deterministic post quantum sEUF-CMA secure signatures are implied by iO). I don't see how having access to a decryption oracle before getting the challenge cipher is in any way helpful to the adversary, as whatever data they get from it is independent of $r_i$.
        
        The only thing here I am concerned about is the move from iO to eO (which as I recall is the same as diO). I need to verify that the reduction gives me the properties I think it does. However, this step is very similar to an argument made by Zhandry in his security proof for signature tokens from iO}
        
        \item Some questions: classical decryptor generation, revocability, what security can we gain from SSL rather than QCP.
        
\end{enumerate}
}

%auto-ignore
\ifnum\anon=0
        \ifnum\sigconf=0
            \author[1]{Or Sattath}
            \author[1,2]{Shai Wyborski}
            \affil[1]{Computer Science Department, Ben-Gurion University}
            \affil[2]{School of Computer Science and Engineering, The Hebrew University of Jerusalem}
        \else
            \author{Or Sattath}
            \affiliation{%
            \institution{Computer Science Department, Ben-Gurion University}
            \country{Israel}}
        \fi
\else
    \ifnum\sigconf=0
        \author{}
    \fi
\fi

\ifnum\sigconf=0
    \maketitle
\fi

%auto-ignore
\begin{abstract}
\emph{Uncloneable decryptors} are encryption schemes (with classical plaintexts and ciphertexts) with the added functionality of deriving uncloneable quantum states, called \emph{decryptors}, which could be used to decrypt ciphers without knowledge of the secret key %(Georgiou and Zhandry, IACR'20)%
\cite{GZ20}. We study uncloneable decryptors in the computational setting and provide increasingly strong security notions which extend the various indistinguishable security notions of symmetric encryption.

We show that \textsf{CPA} secure uncloneable \emph{bit} decryptors could be instantiated from a copy protection scheme (%(Aaronson, CCC'09)%
\cite{Aar09}) for any balanced binary function. We introduce a new notion of \emph{flip detection security} for copy protection schemes inspired by the notions of left or right security for encryption schemes, and show that it could be used to instantiate \textsf{CPA} secure uncloneable decryptors for messages of unrestricted length.

We then show how to strengthen the \textsf{CPA} security of uncloneable decryptors to \textsf{CCA2} security using strong \textsf{EUF-CMA} secure digital signatures. We show that our constructions could be instantiated relative to either the quantum oracle used in \cite{Aar09} or the classical oracle used in \cite{ALLZZ20} to instantiate copy protection schemes. Our constructions are the first to achieve \text{CPA} or \text{CCA2} security in the symmetric setting.
\end{abstract}

\ifnum\sigconf=1
    \keywords{}
    \maketitle
\fi
\newpage
\setstretch{0.9}
\setcounter{tocdepth}{2}
\tableofcontents
\singlespacing

\newpage
%auto-ignore
\section{Introduction} % (fold)
\label{sec:introduction}
Consider a content provider who wishes to broadcast her content over satellite. How could she only allow paying customers to access their content? Classically, the broadcast data is encrypted, and each customer is furnished with a set-top box that contains the decryption key. This approach protects the data from eavesdroppers, but it has an obvious flaw: a hacker with access to just one set-top box could extract the secret key and create as many set-top boxes as they desire. Indeed, current solutions rely on either security by obscurity or authenticating the user, which requires two-way communication (and are therefore unsuitable for satellite broadcast communication).

\emph{Uncloneable decryptors encryption schemes} (or simply \emph{uncloneable decryptors}) solve this problem. Rather than storing a secret key, the set-top box stores a quantum decryptor. The uncloneablility of the decryptors prevents a pirate, even one with access to many set-top boxes, from creating a larger amount of boxes that could decode the broadcast content.

The notion of uncloneable decryptors makes sense for symmetric and asymmetric encryption schemes. In this work we concentrate on the former.

A central component in our constructions is \emph{quantum copy protection} \cite{Aar09,ALLZZ20}. Quantum copy protection is the practice of compiling a function into a quantum state called the \emph{copy protected program} which is useful to evaluate the function. The security of such schemes is defined to prohibit any pirates with access to $n$ copies of the program from creating $n+1$ programs of their own, such that each program could be used to evaluate the program on a given input with a non-negligible advantage over an adversary which outputs a uniformly random guess (the security is defined with respect to the distribution the function and input are sampled from). We formally introduce and discuss this notion in \cref{sec:qcp}.

\subsection{Results}

In this work, we study \emph{uncloneable decryptors}, a primitive that augments symmetric encryption with the functionality of deriving uncloneable quantum states -- which we call \emph{decryptors} -- that can be used to decrypt encrypted messages without the secret key.

The study of uncloneable decryption was initiated in \cite{GZ20} (which we survey in more detail in \cref{ssec:int_rel}) where the notion of \emph{single} decryptor schemes was first defined. We extend the syntax to allow deriving arbitrarily many decryptors, which justifies the choice to change the name.

We mostly consider the computational setting, where the adversary can access various decryptors and may make encryption and decryption calls. We define notions of security which extend the notions of \sindqcpa, \sindqcca and \sindqccaa of \cite{BZ13} (which we formally introduce in \cref{ssec:pre_ses}). These notions extend indistinguishability security of encryption schemes to the post-quantum setting by allowing the adversary to make oracle calls in superposition (as we precisely define in \cref{ssec:oracles}). However, the challenge phase (in which the adversary provides two plaintexts $m_0,m_1$ and has to distinguish their ciphertexts) is done classically. 

We introduce our security notions of \sud, \sudqcpa, \sudqcca and \sudqccaa in \cref{def:dcs-ind,def:dcs-indx}, and in \cref{def:dcs-obindx} we present a variant of these definitions appropriate for \emph{single-bit encryption} schemes which we denote \sobud, \sobudqcpa, \sobudqcca and \sobudqccaa respectively. In \cref{ssec:dcs_assoc} we prove that these security notions extend the corresponding notions of \cite{BZ13}.

We also study the implications of quantum copy protection \cite{Aar09} towards uncloneable decryptors. We introduce and discuss the relevant aspects of quantum copy protection in \cref{sec:qcp}. In particular, we show that the security notions of \cite{Aar09} are not sufficiently strong for our construction, as they do not prohibit \emph{splitting attacks}, that is, splitting the copy protected program into two \emph{partial programs}, where each program is useful to evaluate the function on a large fraction of the inputs. In \cref{sssec:qcp_split} we present an explicit attack that splits the program into two partial programs, each able to evaluate the underlying function in half of the domain. As we show in \cref{thm:ext_fail}, such attacks are detrimental to our constructions. In order to overcome this difficulty, we present in \cref{ssec:qcp_ids} the notion of \emph{flip detection security} -- a novel security notion for copy protection schemes, inspired by notions of "left or right" security for encryption schemes (as we more thoroughly discuss in \cref{app:lor_qcp}), which is secure against such attacks. We use \swqcp and \sfdqcp to indicate the security notions of \cite{Aar09} and \cref{ssec:qcp_ids}, respectively.

All of the security notions above are first defined in the context where $n$ decryptors (or copy-protected programs) are given to an adversary who attempts to create $k$ \emph{additional} copies. We indicate this by putting $n,k$ in the superscript, e.g. \sudqccank or $\swqcp^{1,1}$.

With the syntax and security of uncloneable decryption and copy protection established, we prove the following results:
\begin{itemize}
    \item In \cref{con:dcs_cca1_res} (\cref{ssec:con_ud1_qcpa}) we instantiate an uncloneable bit decryptors scheme from a quantum copy protection scheme for any balanced binary function. In \cref{thm:con_1bit_dcs_cca1} we prove this construction is \sobudqcpank as long as it is instantiated from a \swqcpnk secure copy protection scheme. 
    \item In \cref{ssec:dcs_assoc} we show that \sobudqx security implies \sindqx security for the underlying symmetric encryption scheme. Together with the construction of the previous item, this allows us to infer two interesting properties of \swqcp security: 1. the existence of a $\swqcp^{1,1}$ secure copy protection scheme for a balanced binary function implies the existence of post-quantum one-way functions (\cref{cor:owf}), and 2. any \swqcp copy-protectable \bbf must be a weak $\prf$ (\cref{cor:prf}).
    \item In \cref{con:dcs_cca2_res} (\cref{sssec:obindccaa}) we transform a \sobudqcpank secure scheme to a \sobudqccaank secure scheme by means of post-quantum \sseufcma secure digital signatures (which are discussed in \cref{ssec:ds}). We prove the security of this transformation in \cref{thm:dcs_cca2_res}.
    \item In \cref{sssec:udqcpa} we discuss the security of extending an uncloneable decryptors bit encryption scheme to a scheme that supports plaintexts of arbitrary length by a simple "bit by bit extension" transformation (which we formally define in \cref{def:extend}). We show in \cref{thm:ext_fail} an example of a bit encryption scheme which is \sobudqccaank secure (assuming the existence of a \swqcpnk secure copy protection scheme), whose bit by bit extension to support two bit long messages fails to admit even $\sud^{1,1}$ security. In contrast, in \cref{thm:dcs_cca1} we show that if we instantiate \cref{con:dcs_cca1_res} with a \sfdqcpnk secure copy protection scheme, then extending this scheme bit by bit results in a \sudqcpank secure uncloneable decryptors scheme.
    \item In \cref{con:dcs_cca2} we transform any uncloneable bit decryptor scheme whose bit by bit extension is \sudqcpank secure to a \sudqccaank secure scheme by means of post-quantum \sseufcma secure digital signatures. This construction is a generalization of \cref{con:dcs_cca2_res}. We prove the security of this construction in \cref{thm:udccaank}.
    \item We prove that unconditional \sud security is impossible against an adversary with access to arbitrarily polynomially many decryptors, even if they are not given access to any amount of ciphertexts. This result is incomparable to the impossibility result of \cite{GZ20} as our adversary has access to more decryptors but no ciphers. For more details see \cref{ssec:dcs_uncond}.
    \item In \cref{app:lor_qcp} we show that a \sindqccaa could be instantiated with respect to a quantum oracle, and a $\sindqccaa^{1,1}$ secure scheme could be instantiated relative to a classical oracle. This construction improves upon the previously known constructions, which achieve at most $\sindqcpa^{1,1}$ security.
\end{itemize}

\subsection{Related Works}\label{ssec:int_rel}

\paragraph{Uncloneability in encryption} Uncloneability first appeared in the context of encryption schemes in \cite{Got03}, who presents a scheme for encrypting classical data into a quantum cipher. In their setting, it is possible to split the cipher into two states from which the original plaintext could be recovered upon learning the key, but the fact that the cipher was split could always be \emph{detected} by the intended receiver.

An encryption scheme with uncloneable \emph{ciphers} is presented in \cite{BL20}. In this setting, the receiver is given a quantum ciphertext for a classical message that could be decrypted later, given the classical secret key. The security requirement of this primitive is that an adversary given the ciphertext could not split it into two states, both of which could be used independently to recover the plaintext given the secret key. They provide a construction based on conjugate coding and prove that it is unconditionally secure in the QROM against adversaries with access to a single cipher, provided that the holders of the two parts of the split cipher do not share entanglement.

The study of uncloneable decryptors was initiated in \cite{GZ20} and was studied therein in the context of a \emph{single} decryptor. The authors formalize the notion of single decryptor schemes in private and public key settings, and the security thereof. They then proceed to show a black box transformation of the uncloneable cipher scheme of \cite{BL20} to a single decryptor scheme that is \emph{selectively} secure (that is, the adversary has to pick the messages before seeing the decryption key). The black box properties of this transformation imply that the resulting scheme inherits the security properties of \cite{BL20}'s scheme; namely, it is also secure in the QROM against unentangled adversaries with access to a single cipher. The authors of \cite{GZ20} also construct computationally secure public-key single decryptor schemes based on one-shot signatures and extractable witness encryption. This construction is secure against adversaries with access to a single cipher in the common reference string model. Note that these assumptions are quite strong, in particular, the only known candidate for one-shot signatures is instantiated with respect to a classical oracle \cite{AGKZ20}.

The authors of \cite{CLLZ21} use hidden coset states to present two constructions for public key single decryptor schemes which are secure in the plain model with respect to some assumptions: the first construction requires extractable witness encryption, and the second removes this requirement assuming post-quantum compute-and-compare obfuscation for the class of unpredictable functions.

\paragraph{Quantum copy protection} Quantum copy protection was first presented in \cite{Aar09}, which constructs a copy protection scheme for any quantum-unlearnable function relative to the existence of a quantum oracle. This construction is secure against QPT adversaries with access to arbitrarily many copies of the copy protected program. The author also provides several candidate constructions for copy protection schemes of delta functions. The oracle construction of \cite{Aar09} is the only one that is known to be secure in the presence of multiple copy-protected programs.

The authors of \cite{ALLZZ20} point out some weaknesses in the original security definition of \cite{Aar09} and propose a strictly stronger definition. In particular, they briefly discuss the possibility of splitting attacks such as the one we formalize in \cref{sssec:qcp_split}. Their definition roughly relies on the idea that instead of testing a program against a sampled function and input pair, they could quantify the "quality" of the program against the entire distribution. They also replace the quantum oracle of \cite{Aar09} with a \emph{classical} oracle and prove that relative to this oracle all quantum-unlearnable functions are securely copy protectable even with respect to their stronger notion of security. However, this construction is only known to be secure in the presence of a single copy. 

The authors of \cite{BJLPS21} construct a single copy "semi-secure" copy protection scheme for compute-and-compare programs where it is assumed that at least one of the freeloaders is honest. The authors of \cite{CMP20} construct a copy protection scheme for compute-and-compare circuits without any assumptions which provably admits non-trivial single copy security (that is, it does not completely satisfy the security definition of copy protection, but it does satisfy that the probability that both freeloaders evaluate the function correctly is non-negligibly smaller than $1$, which is classically impossible) in the quantum random oracle model.

In \cite{CLLZ21} it is shown how to construct a copy protection scheme for $\prf$s from the two constructions for uncloneable decryptors mentioned earlier. Their constructions are single-copy secure relying on the same assumptions (either extractable witness encryption or post-quantum compute-and-compare obfuscation for the class of unpredictable functions) along with post-quantum one-way functions with subexponential security and indistinguishability obfuscation. This is the only example of a provably secure copy protection scheme for a non-evasive class of functions.

The authors of \cite{AP20} initiate the study of a weaker form of copy protection called \emph{secure software leasing}. In this notion, it is assumed that the end users of the program are honest. That is, piracy is not preventable, but it is always detectable. They construct such a scheme for evasive functions based on indistinguishability obfuscators and LWE. They further provide a construction for a class of quantum unlearnable functions assuming LWE, but for which an SSL scheme could not be constructed. This implies that a general scheme for quantum copy protection for quantum unlearnable functions is also impossible. More recently, \cite{KNY20} managed to construct SSL schemes for a subclass of the class of evasive functions. They also consider two weaker variants of SSL --  \emph{finite-term SSL} and \emph{SSL with classical communications} -- and show constructions of such schemes for pseudo-random functions. Their constructions are single copy secure in the common reference string model.

\paragraph{Post-quantum indistinguishability of ciphers}
Security notions for classical encryption and digital signature schemes against quantum adversaries were first considered in \cite{BZ13}, where the notions of \sindqcca and \sindqccaa security were introduced. These notions allow the adversaries to query the encryption/decryption oracles in superposition but do not allow making the challenge query in superposition. The authors show that two seemingly natural security definitions that allow superimposed challenge queries are actually impossible to obtain and leave open the problem of meaningfully defining \sqindqcpa security.

The authors of \cite{GHS16} propose a definition for \sqindqcpa and prove that it extends the \sindqcpa notion of \cite{BZ13}. They show that this notion is not achievable by any quasi-length-preserving encryption schemes but provide a length-increasing scheme that is \sqindqcpa secure.

In \cite{CEV20}, the authors propose a definition for \sqindqccaa and prove that it is stronger than the \sqindqcpa notion of \cite{GHS16}. They provide constructions for \sqindqcca secure symmetric encryption schemes and prove that the encrypt-then-MAC paradigm affords a generic transformation of \sqindqcca secure schemes to \sqindqccaa secure schemes.

\subsection{Scientific Contribution}

We improve upon existing work in the following senses:
\begin{itemize}
    \item our constructions are the first to exhibit security in the presence of multiple decryptors,
    \item our work is the first to obtain \text{CPA} and \text{CCA2} in the symmetric setting, and
    \item to the best of our knowledge, this is the first work to exhibit an application of copy protection to cryptography.
\end{itemize}

\subsection{Drawbacks}

The main drawback of our constructions is that they currently cannot be instantiated in any standard model to obtain an uncloneable decryptors scheme which is not a single decryptor scheme. Our constructions require quantum copy protection of balanced binary functions. The only known construction of such a scheme is the one presented in \cite{CLLZ21}, which relies on assumptions that \cite{CLLZ21} also use to construct single decryptor schemes. However, as we discuss in \cref{sssec:fdqcp_oracle}, our schemes can be instantiated relative to a quantum oracle. Furthermore, this oracle could be replaced with a classical oracle at the cost of the resulting scheme only being secure as a single decryptor scheme.

In order to obtain security against messages of arbitrary length, our construction assumes flip detection security. It is yet unclear if flip detection security is obtainable from weak security \cite{Aar09}, or if it is implied by strong copy protection \cite{ALLZZ20}. (However, we show in \cref{app:flip_unlearn} that the oracle \emph{constructions} of \cite{Aar09} and \cite{ALLZZ20} satisfy flip detection security.)

The security notions presented below are given as indistinguishability experiments. Our notions would be better established had we provided an equivalent semantic (simulation based) security definition. At this point, we do not even have a candidate notion of semantic security. We leave this as an open problem.

\subsection{Overview of the Constructions}

We give a short informal description of our constructions. We use the syntax for uncloneable decryptors, which we formally define in \cref{def:dcs} (as well as syntaxes for copy protection and digital signatures). However, we hope that the function of the procedures we refer to is sufficiently clear from context for the purpose of this exposition.

Let \sqcp be a copy protection scheme for a balanced binary function \bbf with input length $\ell$. We transform \sqcp to an uncloneable bit decryptor scheme using a transformation similar to the standard transformation of pseudo-random functions to symmetric encryption schemes (see e.g. \cite[Construction~5.3.9]{Gol04}):

\begin{itemize}
    \item The secret key $\sk$ is an (efficient description of) a function sampled from \bbf, that is, $\sk\gets \bbf.\samp(\secparam)$.
    \item To generate a decryptor, copy protect the function: $\decgen(\sk) \equiv \qcp.\ptect(\sk)$.
    \item To encrypt a bit $b$, sample $r\gets \zo^\ell$ and output the cipher $(r,b\oplus \bbf.\eval_\sk(r))$.
    \item To decrypt a cipher $(r,\beta)$ using a decryptor $\rho$ calculate $b=\beta \oplus \qcp.\qeval(\rho,r)$.
\end{itemize}

To obtain \sobudqccaa security from a \sobudqcpa secure scheme \ud, we wrap it with a digital signature scheme \ds:

\begin{itemize}
    \item Sample a key pair $(\sk_\ds,\pk_\ds)\gets \ds.\keygen(\secparam)$ and attach $\sk_\ds$ to the secret key and $\pk_\ds$ to any generated decryptor.
    \item To encrypt a bit $b$ generate $c\gets \ud.\enc(b)$ and output $(c,s=\ds.\sign_{\pk_\ds}(c))$.
    \item To decrypt a cipher $(c,s)$ first verify that $\ds_{\pk_\ds}(c,s)=1$ and if the verification passes use \ud to decrypt the message.
\end{itemize}

In order to extend the first scheme to support arbitrary length plaintexts, we encrypt each bit separately. This construction remains secure provided that the underlying copy protection scheme \qcp has sufficiently strong security.

To obtain \sudqccaa security we adjust the transformation we used for bit decryptors by:
\begin{itemize}
    \item for each cipher, sampling a random \emph{serial number} $r\gets\zo^\secpar$, and
    \item signing the cipher of each bit along with the serial number, the plaintext length, and the index of the bit (that is, if the cipher of the \sudqcpa scheme is $c_1,\ldots,c_\ell$ we attach to $c_i$ the signature $s_i = \ds.\sign_{\sk_\ds}(c_i,i,\ell,r)$).
\end{itemize} Signing $\ell$ prevents an adversary from truncating ciphers, signing $i$ prevents rearranging the bits of a cipher, and signing $r$ prevents splicing ciphers.

\ifnum\anon=0
\subsection{Acknowledgements}
We wish to thank Dominique Unruh for valuable discussions. This work was supported by the Israel Science Foundation (ISF) grant No. 682/18 and 2137/19, and by the Cyber Security Research Center at Ben-Gurion University.
\fi
%auto-ignore
\section{Preliminaries}\label{sec:pre}
\subsection{Basic Notions, Notations and Conventions} \label{ssec:pre_notations}

We call a function $f:\mathbb{R}\to\mathbb{R}$ \emph{negligible} if for every $c<0$ it holds for any sufficiently large $\secpar$ that $|f(\secpar)| < \secpar^c$. Equivalently, $f$ is negligible if $|f(\secpar)|=\secpar^{-\omega(1)}$. We often use the shorthand $f(\secpar) \le \negl$ to state that $f$ is negligible, and the shorthand $f(\secpar) \le g(\secpar) + \negl$ to state that $|f-g|$ is negligible.

We say that $f:\mathbb{N}\to \mathbb{N}$ is \emph{polynomial in \secpar} and denote $f=\poly$ if there exists a polynomial $p$ such that $f(\secpar) \le p(\secpar)$ for all $\secpar \in \mathbb{N}$. Note that we do not require that $f$ is a polynomial, only that it is upper bound by a polynomial.

A \emph{quantum polynomial time (QPT)} procedure $C$ is a uniform family of circuits $C_\secpar$ such that $|C_\secpar| = \poly$, where we use $|C_\secpar|$ to denote the number of elementary gates in $C_\lambda$. We always assume that the elementary gates are chosen from a fixed finite universal set. Equivalently, $C$ is QPT if there exists a polynomial Turing machine whose output on the input $\secparam$ is a classical description of the circuit $C_\secpar$.

We use sans serif typeface to denote circuits whose input and output are classical such as $\enc$ and classical data such as $\sk$. We use calligraphic typeface to denote quantum algorithms and data such as $\qenc$ and $\qsk$. We use Greek letters $\rho,\sigma,\ldots$ to represent (possibly mixed) quantum states given as density operators. When $\rho$ is a state and $s$ is a classical string we often abuse notation by writing $\rho\otimes s$ as shorthand for $\rho\otimes \ket{s}\bra{s}$.

Sans serif typeface does \emph{not} indicate that the circuits are not quantum, only that they map computational basis states to computational basis states. When the circuits are random, we always assume that the random bits are uniformly sampled from the computational basis. This is further elaborated upon in \cref{ssec:oracles}.

We use sans serif uppercase letters to denote instances of (either classical or quantum) schemes, such as $\se$ to denote symmetric encryption schemes. When there are several schemes in play, we will use notations such as $\se.\keygen$ to clarify to which scheme we refer. 

When $\sk$ denotes a key, we will often use currying notation such as $\enc_\sk$ rather than $\enc(\sk,\cdot)$ for clarity of exposition. We emphasize that $\sk$ is still considered part of the input to $\enc$.

Given a function $f:D\to R$ and $S\subset D$, we define the function
$$(f\setminus S)(x) = \begin{cases} f(x) & x\notin S \\ \bot & x\in S \end{cases}\text{,}$$ if $c\in D$ we slightly abuse notation by writing $f\setminus c$ instead of $f\setminus \{c\}$.

For any string $c\in \{0,1\}^*$ we  use $|c|$ to denote its length and $\overline{c}$ to denote its bitwise flip.

For any distribution $\sD$ we use the notation $x\gets \sD$ to indicate that $x$ is sampled from the distribution $\sD$. For any finite set $A$ we use the notation $x\gets A$ to indicate that $x$ is sampled from $A$ \emph{uniformly}.

\subsection{Classical Oracles with Quantum Access}\label{ssec:oracles}

Some of the circuits we discuss below are implicitly \emph{oracle machines}. This means that they are given access to query an oracle (that is, a black box that evaluates some function). Given an oracle machine $\mathcal{A}$ and a \emph{classical} function $f$, we use the notation $\mathcal{A}^{f}$ to denote an invocation of $\mathcal{A}$ when it is given oracle access to $f$.

In our context, we wish to model a quantum adversary which can perform a classical circuit on an input which is \emph{in superposition}. Even though the underlying function is \emph{classical}, the adversary can access it \emph{quantumly}. We now precisely define what this means.

For deterministic oracles, quantum access means that the adversaries may query the oracle superposition as follows: the standard way to represent a function $f$ as a quantum operation as the unitary $U_f$ defined via $\ket{x,y}\mapsto \ket{x,y\oplus f(x)}$. Given an oracle machine $\sA$ and a classical function $f$ we use $\sA^{\ket{f}}$ to indicate that $\sA$ my perform quantum queries to the unitary $U_f$. We call this \emph{quantum access} to the (classical) function $f$.

Some of the oracles we consider are randomized (e.g. encryption oracles). In case $C$ is some randomized circuit, we use $C(\cdot;r)$ to denote the (deterministic) invocation of $C$ with random bits $r$. We define the quantum oracle of $C$ as the oracle which, upon receiving a query (possibly in superposition), samples a uniformly random $r$ and applies to the state the unitary $U_{C;r}$ defined by $\ket{m,x} \mapsto \ket{m,C(m;r)\oplus x}$. Given an oracle machine $\sA$ and a randomized circuit $C$ we use $\sA^{\ket{C}}$ to indicate that $\sA$ may perform quantum queries to the oracle which uniformly samples $r$ and then applies $U_{C;r}$ to the input. Note that the randomness is sampled once \emph{per query}.

\begin{remark}
    This oracle could be equivalently defined as the oracle which, on input $\rho$, applies the unitary $\ket{m,x,r} \mapsto \ket{m,C(m;r)\oplus x,r}$ to $\rho\otimes\left(2^{-|r|}\sum_r\ket{r}\bra{r}\right)$.
\end{remark}

\subsection{Symmetric Encryption Schemes}\label{ssec:pre_ses}

Recall the syntax of symmetric encryption schemes:

\begin{definition}[Symmetric Encryption Scheme]\label{def:ses}
    A \emph{symmetric encryption} (or \emph{private key encryption}) \emph{scheme} $\se$ is a triplet of QPT circuits:
    \begin{itemize}
        \item $\sk \leftarrow \keygen(1^\lambda) $,
        \item $c \leftarrow \enc_\sk(m) $, and
        \item $m  \leftarrow \dec_\sk(c)$.
    \end{itemize}
    $\se$ is \emph{correct} if for all messages $m\in \{0,1\}^*$ it holds that
    $$\prob{
            \sk\leftarrow\keygen\left(1^{\lambda}\right)
        :
            \dec_\sk(\enc_\sk(m)) = m
    } = 1 $$
\end{definition}

Security of encryption schemes is commonly defined in terms of \emph{indistinguishability games} between a trusted challenger $\sC$ and an arbitrary adversary $\sA$ which takes the following form:
\begin{enumerate}
    \item (Initialization) $\sC$ invokes $\keygen(\secparam)$ to obtain a secret key $\sk$ which it provides to $\sA$.
    \item (First learning phase) $\sA$ produces two messages $m_0$, $m_1$.
    \item (Challenge query) $\sA$ provides $\sC$ with $m_0$, $m_1$, $\sC$ samples a uniformly random $b\gets\{0,1\}$ and responds with the \emph{challenge cipher} $c=\enc_\sk(m_b)$.
    \item (Second learning phase) $\sA$ outputs a single bit $b'$.
    \item (Winning condition) The output of the game is $1$ if $b'=b$.
\end{enumerate}
The scheme is considered secure if any QPT $\sA$ can win the game with probability at most $1/2 + \negl$.

We define increasingly strong security notions by considering increasingly strong adversaries, where we model the adversary's strength by the oracle access they are allowed to have. The most common notions are:
\begin{itemize}
    \item Passive Attack (\sind): $\sA$ is not given access to any oracle.
    \item Chosen Plaintext Attack (\sindcpa): $\sA$ is given classical oracle access to $\enc_\sk$.
    \item A priori Chosen Ciphertext Attack (\sindcca): $\sA$ is additionally given classical oracle access to $\dec_\sk$ \emph{during the first learning phase}.
    \item A posteriori Chosen Ciphertext Attack (\sindccaa): $\sA$ is additionally given classical oracle access to $\dec_\sk\setminus c$ during the second learning phase, where $c$ is the challenge cipher.
\end{itemize}

Another way to increase the adversary's strength is by allowing her to make oracle queries in superposition (as described in \cref{ssec:oracles}). We denote the result of modifying the \sindx game this way by \sindqx. Note that in the \sindqx indistinguishability games the challenge query is still completely classical. This notion of security was first explored in \cite{BZ13}. For completeness, we provide the full definition of \sindqcca security. \sindqcpa and \sindqccaa security are defined by appropriately modifying the oracle access of $\sA_1$ and $\sA_2$.

\begin{definition}[{\sindqcca security, \cite[Definition~8]{BZ13}}] \label{def:ind-qccax}
    Let $\se$ be a symmetric encryption scheme (cf. \cref{def:ses}). For a procedure $\sA = (\sA_1,\sA_2)$ let the $\sindqcca_{\se}(\sA,\lambda)$ game between $\sA$ and a trusted challenger $\sC$ be defined as follows:
    \begin{enumerate}
        \item $\sC$ invokes $\keygen(\secparam)$ to obtain the key $\sk$,
        \item $\sA_1^{\ket{\enc_\sk}\ket{\dec_\sk}}$ outputs two classical plaintexts $m_0,m_1$ and auxiliary data $\sigma$,
        \item $\sC$ samples $b\gets \zo$ and computes $c\gets\enc_\sk(m_b)$,
        \item the output of the game is $1$ if $\sA_2^{\ket{\enc_\sk}}(\sigma,m_0,m_1,c)=b$\label{ent:ind-qcca1}
    \end{enumerate}
    
    $\se$ is \emph{\sindqcca secure} if for any QPT $\sA$ it holds that
    $$\PP[\sindqcca (\sA,\lambda)=1] = \frac{1}{2} + \negl\text{.}$$
\end{definition} 

\begin{remark}
    One can also attempt allowing superimposed \emph{challenge queries}, arriving at notions of the form \sqindqx. However, it seems that the most straightforward ways to try to do so lead to security notions that are impossible to obtain, see e.g. \cite[Thoerem~4.2 and Theorem~4.4]{BZ13}. Recent works such as \cite{GHS16,CEV20} provide meaningful definitions for such security notion, as further discussed in \cref{ssec:int_rel}.
\end{remark}

\subsection{Digital Signature Schemes}\label{ssec:ds}

\begin{definition}[Digital Signature Scheme]A \emph{digital signature scheme} $\ds$ is a triplet of QPT circuits:\label{def:ds}
\begin{itemize}
    \item $\sk,\pk \gets \keygen(1^\lambda)$,
    \item $s\leftarrow\sign_\sk(m)$, and
    \item $b\leftarrow\ver_\pk(m,s)$.
\end{itemize}
where $\ver$ is deterministic.
$\ds$ is \emph{correct} if for all messages $m$ it holds that
    $$\PP\left[
        b=1
        :
        \begin{matrix}
            \sk,\pk\leftarrow\keygen(1^\lambda) \\
            s\leftarrow\sign_\sk(m) \\
            b\leftarrow\ver_\pk(m,s)
        \end{matrix}
        \ 
    \right] = 1{.}$$
    We say that \ds is \emph{deterministic} if $\ds.\sign$ is deterministic.
\end{definition}

The notion of security we require for digital signature is that of \emph{strong existential unforgeability under chosen message attacks}, which we abbreviate as \sseufcma. Under this notion, we furnish an adversary with access to a signing oracle as well as the public verification key and expect him to create a signed document $(m,s)$ such that $\ds.\ver_{\pk}(m,s) = 1$ though $s$ was not output as a response to a signature query on $m$. If this holds for any QPT adversary, we say that $\ds$ is (post-quantum) \sseufcma secure. 

\begin{definition}[\sseufcma security (adapted from \cite{GMR88})] \label{def:euf-qcma}
    Let $\ds$ be a digital signature scheme as in \cref{def:ds}. For a procedure $\sA$ let the strong $\sseufcma_\ds(\sA,\lambda)$ game be defined as follows:
    \begin{enumerate}
        \item $\sC$ generates $\pk,\sk\gets \keygen$ and gives $\pk$ to $\sA$.
        \item $\sA^{\sign_\sk}$ produces a signed document $(m,s)$.
        \item The result of the game is $1$ if $\ver_{\pk}(m,s) = 1$ and $s$ was not given as output to a query with input $m$ during the previous phases.
    \end{enumerate}
    $\ds$ is \emph{strong \sseufcma secure} if for any QPT $\sA$ it holds that
    $$\PP[\sseufcma_{\ds}(\sA,\lambda)=1] \le \negl\text{.}$$
\end{definition}

In the \emph{weak} variant $\mathsf{EUF\mhyphen CMA}$, we only require that the signing oracle is never queried on $m$. The resulting security notion does not prohibit attacks where a signed document $(m,s)$ could be used to create $s'\ne s$ such that $\ds.\ver_\pk(m,s')=1$. Stated differently, this security prevents an adversary from signing previously unsigned messages but not from creating new valid signatures for a previously signed message. As we elaborate in \cref{sssec:obindccaa}, we use signature schemes to make it unfeasible to create new valid ciphers from a list of ciphers for chosen plaintexts. If the adversary can modify the signature without abrogating its validity, they can transform a known plaintext into a new plaintext, which we expected the signature scheme to protect us from. Hence, $\mathsf{EUF\mhyphen CMA}$ is unsuitable for our application.

The notion of $\mathsf{EUF\mhyphen CMA}$ security first appeared in \cite{GMR88}, though the authors thereof only required security against PPT adversaries. By \emph{post-quantum security} we mean the same security game, but where the adversary is QPT rather than PPT. However, we still require that the adversary only has classical access to the signing oracle.

\begin{remark}
    Extending this notion to an adversary who is allowed to make signature queries in superposition is not at all straightforward, since in this case it is impossible to record the queries made by the adversary, and it is unclear even what it means for a signed document to be "different" than the responses to the queries they made. This has been addressed e.g. in \cite{BZ13}, where they do not record the queries but rather require that an adversary which makes $q$ queries can not create $q+1$ distinct signed messages that all pass the verification procedure.
\end{remark}

%auto-ignore
\section{Quantum Copy Protection}\label{sec:qcp}

In this section we overview the notion of \emph{quantum copy protection} -- the practice of compiling an arbitrary functionality into a quantum program in a manner that makes the functionality uncloneable. Copy protection was first introduced and discussed in \cite{Aar09}, which furnishes a first attempt at a security definition. For reasons which will become clear shortly, we refer to this security notion as \emph{weak} copy protection, which we discuss at some length in \cref{sssec:wqcp}. While this security notion is far from trivial, it also exhibits some vulnerabilities which make it unsuitable for many cryptographic applications. One such vulnerability is that it does not prohibit \emph{splitting} the program into two "partial" programs, each able to evaluate the protected function on a different portion of the domain. We exhibit an explicit splitting attack in \cref{sssec:qcp_split}. In \cite{ALLZZ20} a much stronger security notion is proposed, which does prohibit splitting attacks. We shortly and informally discuss this notion in \cref{sssec:qcp_new}. We further discuss the state of the art of quantum copy protection \cref{ssec:int_rel}.

Our constructions only require the copy protection of a balanced binary function. That is, an efficiently sampleable distribution of binary functions such that applying a sampled function to a uniform input distributes negligibly close to a uniformly random bit. In \cref{ssec:qcp_ids} we introduce \emph{flip detection security}, a strengthening of weak copy protection which prohibits splitting attacks.

\subsection{The Syntax of Quantum Copy Protection}

\begin{definition}[Admissible Class]\label{def:admiss}
    An \emph{admissible class} of functions $\sF_\secpar$ is a collection of functions with the following properties:
    \begin{itemize}
        \item the members of $\sF_\secpar$ are functions $\zo^\ell \to \zo^m$ where $n,m=\poly$,
        \item there exists $d = \poly$ such that for any $f\in \sF_\secpar$ there exists a string $d_f$ with $|d_f|=d$, this string is called the \emph{description} of $f$, and
        \item there exists a circuit $\eval^\sF$ such that for any $f\in \sF_\secpar$ and $x\in \zo^n$, $\eval^\sF(d_f,x)$ outputs $f(x)$ with running time $\poly$.
    \end{itemize}
\end{definition}

When considering a class $\sF_\secpar$ we usually suppress the security parameter $\secpar$, we also use the notation $\eval_f$ as shorthand for $\eval^\sF(d_f,\cdot)$. We often (e.g. in the definition below) abuse notation and refer to $f$ instead of to $d_f$, and to $\sF$ as the collection of descriptions rather than the collection of functions themselves.

\begin{definition}[Copy Protection Scheme, adapted from \cite{Aar09}]\label{def:qcp_synt}
    Let $\sF$ be an admissible class, a \emph{copy protection scheme} for $\sF$ is a pair of procedures $\sqcp_\sF = (\ptect,\qeval)$ with the property that if $\sigma_f \gets \ptect(f)$ then it holds for any $x\in\zo^\ell$ that $\qeval(\sigma_f,x) = f(x)$.
\end{definition}

We often refer to $\sigma_f$ as the \emph{copy protected} version of $f$. When $\sF$ is clear from context, we will suppress it and refer to the scheme by $\sqcp$.

\subsection{Weak Copy Protection}\label{sssec:wqcp}

We present the original notion of quantum copy protection presented in \cite{Aar09}. This definition was subsequently strengthened by \cite{ALLZZ20} whose authors also name the strengthened definition therein quantum copy protection. We hence refer to the original definition of \cite{Aar09} as \emph{weak} copy protection to avoid confusion.

Intuitively, copy protection security is defined in terms of a game between a trusted \emph{challenger} $\sC$ and several arbitrary QPT algorithms, namely a \emph{pirate} $\sP$ and $n+k$ non communicating \emph{freeloaders} $\sF_1,\ldots,\sF_{n+k}$. The pirate is given $n$ copies of the copy protected function $\sigma_f$, from which the create $n+k$ pirated copies $\rho_1,\ldots,\rho_{n+k}$, affording $\rho_j$ to $\sF_j$. Finally, $\sC$ asks each freeloader to evaluate $f$ on some point, and the output of the game is the number of freeloaders which evaluated $f$ correctly. The scheme is considered secure if for any adversary $\sA = (\sP,\sF_1,\ldots,\sF_{n+k})$ the expected output of the game is at most negligibly higher than $n + k/2^m$. Note that a probability of $n+k/2^m$ could be reached easily: for $j=1,\ldots,n$ the pirate gives a copy of $\sigma_f$ to $\sF_j$, so that the first $n$ freeloaders can answer correctly with certainty; the remaining $k$ freeloaders output a uniformly random value from $\zo^m$, giving each a winning probability of $2^{-m}$.

Formalizing this idea requires specifying the distribution from which the function $f$ and the inputs given to the freeloaders are sampled. Hence, copy protection (in all its forms) is defined with respect to a distribution $\sD$ on the class of functions and on the set of inputs, i.e. over $\sF \times \zo^n$.

\begin{definition}[\swqcp security, adapted from \cite{Aar09}]\label{def:wqcp} 
    Let $\sqcp$ be a copy protection scheme for some admissible class of functions, and let $\sD_\secpar$ be an efficiently samplable distribution on $\sF_\secpar\times \zo^n$. For any natural numbers $n,k$ define the game $\swqcp_\sqcp^{n,k,\sD}(\sA,\secpar)$ between a trusted challenger $\sC$ and an arbitrary adversary $\sA = (\sP,\sF_1,\ldots,\sF_{n+k})$:
    \begin{itemize}
        \item $\sC$ samples $(f,x)\gets \sD_\secpar$ and invokes $\sqcp.\ptect(f)$ $n$ times to obtain $\rho = \rho_f^{\otimes n}$,
        \item $\sP(\rho)$ generates $n+k$ states $\sigma_1,\ldots,\sigma_{n+k}$,
        \item $\sF_i(\sigma_i,x)$ outputs a string $y_i$,
        \item the output of the game is the number of indices $i$ such that $f(x)=y_i$.
    \end{itemize}
    
    $\sqcp$ is \emph{$\swqcp^{n,k,\sD}$ secure} if for any QPT adversary $\sA$ it holds that $$\epsilon = \E[\swqcp^{n,k,\sD}(\sA,\secpar)] \le n + \frac{k}{2^m} + \negl \text{,}$$ $\sqcp$ is \emph{$\swqcp^\sD$ secure} if it is $\swqcp^{n,k,\sD}$ secure for any $n,k=\poly$.
    
    When $\sD$ is clear from context we suppress it and write $\swqcp^{n,k}$ and $\swqcp$.
\end{definition}

As explained in \cite{Aar09}, a copy protection scheme could not be (weakly) secure against arbitrary adversaries, as an unbounded $\sP$ could use $\rho_f$ to learn $f$. A necessary condition for the existence of a copy protection scheme for $f$ is that $f$ is \emph{quantum-unlearnable}, which roughly means that a polynomial quantum adversary with oracle access to $f$ can not evaluate $f$ on a sampled $x$ once the oracle access is revoked (unlearnability is defined with respect to the distribution $f$ and $x$ are sampled from). We expand more on unlearnability in \cref{app:flip_unlearn}. However, quantum unlearnability is not a sufficient condition: \cite{AP20} construct a family of functions which is (under some standard cryptographic assumptions) quantum unlearnable and yet is \emph{not} copy-protectable.

\subsection{Balanced Binary Functions}

A \bbf $\sF$ is comprised of an admissible class (see \cref{def:admiss}) of \emph{binary} functions (that is, $m=1$) and an efficient procedure for sampling $f\gets \sF$ such that if $x$ is uniformly random then $f(x)$ is very close to uniform (where the distribution is taken over both $f$ and $x$).

More formally:

\begin{definition}[balanced binary function (BBF)]\label{def:bbf}
    A \emph{Balanced Binary Function} $\bbf$ is comprised of two QPT procedures $(\samp,\eval)$ such that:
    \begin{itemize}
        \item $\samp(\secparam)$ samples $f$ (not necessarily uniformly) from a set $\sF_\secpar \subset \{0,1\}^{poly(\secpar)}$. We do not require that all strings in $\sF_\secpar$ will be of the same length, only that their length is bounded by some polynomial in $\secpar$,
        \item there exists a polynomial $\ell$ such that if $f\leftarrow\samp(\secparam)$ then $\eval_f$ implements a deterministic function from $\{0,1\}^{\ell(\secpar)}$ to $\{0,1\}$, and
        \item if $f\leftarrow\samp(\secparam)$ and $x\gets \zo^{\ell(\secpar)}$ then $$\left|\prob{\eval_\sk(x) = 0} - \prob{\eval_\sk(x) = 1} \right| = \negl\text{.}$$
    \end{itemize}
\end{definition}

\begin{remark}\label{rem:bbf_dist}
    When considering copy-protection of $\bbf$s, the distribution $\sD$ is always implicitly assumed to be $\bbf.\samp(\secparam) \times \mathcal{U}(\zo^{\ell(\secpar)})$. That is, the function is sampled from the (randomized) circuit $\bbf.\samp$, and the input is uniformly random. This definition has the nice property that the function is sampled independently from the input, making it easy to extend this definition in various ways that require sampling more than one input. We propose one such way in \cref{ssec:qcp_ids}.
\end{remark}

Being a \bbf in itself is not a strong property. A trivial example of a \bbf is the one containing the two constant functions, which is obviously not unlearnable and in particular not copy protectable. It turns out that a necessary condition for a \bbf to be $\swqcp^{1,1}$ copy protectable is that \bbf is a weak \prf. We state and discuss this corollary in \cref{sssec:impl} and provide a proof in \cref{app:ggm}.

\subsection{A Splitting Attack on Weak Copy Protection}\label{sssec:qcp_split}

As mentioned earlier, the security definition in \cref{def:wqcp} has a weakness that is detrimental to our (and arguably other) applications. Namely, it does not prohibit \emph{splitting attacks}. Intuitively, a splitting attack is an efficient way to transform a copy-protected program $\rho_f$ into two states $\sigma_0$ and $\sigma_1$ such that each state $\sigma_b$ is useful to evaluate a non-negligible fraction of inputs. For example, $\sigma_b$ could be used to evaluate any input whose first bit is $b$.

We informally discuss such a splitting attack on any $\swqcp^{1,1}$ secure scheme and defer formal treatment of this attack to \cref{app:split}.

Let \sqcp be a copy protection scheme for a \bbf with input length $\ell$. We can consider a new class $\bbf'$ of input length $\ell+1$ such that every function in $\bbf'$ is given by a pair of (efficient descriptions of) functions $f_0,f_1$ sampled from \bbf. The function described by $(f_0,f_1)$ maps $b\|x$ to $f_b(x)$. Copy protect $\bbf'$ by sampling two functions $f_0,f_1$, copy protecting each to obtain $\rho_0,\rho_1$ and providing $\rho_0\otimes\rho_1$ as the protected function. Hence, to evaluate on the input $b\|x$ one invokes $\qcp.\qeval(\rho_b,x)$ and returns the output.

The copy-protected program could be trivially split into the two programs $\rho_0$, $\rho_1$, where $\rho_b$ could be used to evaluate the function on any input starting with $b$.

\subsection{Strong Copy Protection}\label{sssec:qcp_new}

Splitting attacks noticed by \cite{ALLZZ20}, who have devised a stronger definition of quantum copy protection that prohibits such attacks. The definition thereof is highly involved and requires introducing several notions it relies upon. Since we never use this definition directly, we compromise for an overview of the intuition behind it.

The idea is not to test the pirated programs produced by the pirate on a sampled output, but rather design a binary measurement whose success probability is the same as the probability that a given program (comprised of a quantum state along with "instructions" for using it to evaluate the copy-protected function in the form of a quantum circuit) evaluates correctly a function-input pair sampled from the respective distribution $\sD$. The scheme is then considered $\sqcp^{1,1,\sD}$ secure if it is impossible to transform a copy-protected function $\rho$ into two programs, such that each program passes the measurement with non-negligible security. The extension to $\sqcp^{n,k,\sD}$ and $\sqcp^{n,k,\sD}$ security is similar to the end of \cref{def:wqcp}.

The authors explain how to efficiently approximate a distribution such as above with respect to any efficiently samplable distribution $\sD$. This implies that the game described above could be approximately simulated by an efficient challenger, which is essential for using this definition for security reduction. A tenet of this security definition is that the measurement is performed on each program independently (though the measurement outcomes may not be independent due to entanglement), which prohibits splitting attacks.

Note that while it is possible to implement this measurement efficiently, it is impossible to efficiently determine that it has a negligibly small success probability. This is in analogy to the fact that it is efficient to simulate the $\swqcp$ game against arbitrary QPT adversaries, but it is infeasible to determine whether an arbitrary adversary has a non-negligible advantage.

\subsection{Flip Detection Security}\label{ssec:qcp_ids}

In this section we define the notion of \emph{flip detection security}.

The intuition is that it should be hard for many freeloaders to distinguish a black box that always evaluates the function correctly from a black box that always evaluates it incorrectly. As we discuss in \cref{app:lor_qcp}, this notion is a natural adaptation of the notions of left and right security used to model the security of encryption schemes in the presence of multiple encryptions.

The notion of \sfdqcp security is only defined for copy-protection schemes for $\bbf$s, recall that when considering such scheme we assume that $\sD = \bbf.\samp \otimes \mathcal{U}(\zo^\ell)$ (see \cref{rem:bbf_dist}).

\begin{definition}[\sfdqcp Security]\label{def:fdqcp}
    Let $\sqcp_\bbf$ be a copy protection scheme for a balanced binary function $\bbf$. For any binary function with input length $\ell$ and any bit $b$ let $\sO_{f,b}$ be the oracle which takes no input and outputs $(r,f(r)\oplus b)$ with $r\gets \zo^\ell$.
    
    For any $n,k =\poly$ define the game $\sidqcp_\sqcp^{n,k}(\sA,\secpar)$ between a trusted challenger $\sC$ and an arbitrary adversary $\sA = (\sP,\sF_1,\ldots,\sF_{n+k})$:
    \begin{itemize}
        \item $\sC$ samples $f \gets \bbf.\samp(\secparam)$ and $b_1,\ldots,b_{n+k}\gets \zo$ and invokes $\sqcp.\ptect(f)$ $n$ times to obtain $\rho=\rho_f^{\otimes n}$,
        \item $\sP(\rho)$ creates $n+k$ states $\sigma_1,\ldots,\sigma_{n+k}$,
        \item $\sF_i^{\sO_{f,b_i}}(\sigma_i)$ outputs a bit $b'_i$,
        \item the output of the game is the number of indices $i$ for which $b_i=b'_i$.
    \end{itemize}

    $\sqcp$ is \emph{$\sidqcp^{n,k}$ secure} if for any QPT adversary $\sA$ it holds that $$\E[\sfdqcp^{n,k}_\qcp(\sA,\secpar)] \le n + \frac{k}{2} + \negl\text{,}$$ $\sqcp$ is \emph{\sidqcp secure} if it is $\sidqcp^{n,k}$ secure for any $n,k=\poly$.
\end{definition}

\begin{lemma}\label{lem:fdtow}
    Let $\sqcp$ be a $\sidqcp^{n,k}$ secure scheme, then it is also $\swqcp^{n,k}$ secure.
\end{lemma}

\begin{proof}
    This follows by noting that if we modify the $\sidqcp^{n,k}$ such that each freeloader makes exactly one query to $\sO_{f,b}$, we obtain a notion which is equivalent to \swqcpnk. The only difference is that instead of just getting $x$, $\sF_i$ gets a pair of the form $(x,f(x)\oplus b_i)$ where $b_i$ is uniformly random. But that $b_i$ is uniformly random implies that $f(x)\oplus b_i$ is also uniformly random, so it could be omitted.
\end{proof}

The splitting attack elaborated in \cref{app:split} implies that the converse of \cref{lem:fdtow} is false, we prove this in \cref{lem:fdfail}.

\cref{thm:ext_fail} shows that this property carries over to uncloneable decryptors: there could exist uncloneable bit decryptors which are secure, but trying to extend them to arbitrary message lengths by encrypting bit by bit is not secure. Fortunately, strengthening the security of copy protection to resemble \slor security rather than \sind security alleviates this problem, as we shall see in \cref{thm:dcs_cca1}.

It is unclear whether strong copy protection implies flip detection or whether it is possible to generically transform a \swqcp secure scheme to a \sfdqcp secure scheme. We leave this as an open question.

\subsubsection{Oracle Instantiation of \textsf{FLIP-QCP} Secure Schemes}\label{sssec:fdqcp_oracle}

In \cref{thm:oracle_inst} we establish that \sfdqcp secure (resp. $\sfdqcp^{1,1}$) secure copy protection schemes exist for any unlearnable \bbf relative to a quantum (resp. classical) oracle.

The author of \cite{Aar09} presents a weakly secure copy protection scheme instantiated relative to a quantum oracle. This scheme is unique in the sense that it supports an arbitrary polynomial amount of copies. The authors of \cite{ALLZZ20} manage to replace this oracle with a classical oracle. However, the resulting construction is not known to support more than a single copy.

In \cref{app:flip_unlearn} we prove that the \cite{Aar09} and \cite{ALLZZ20} schemes are in fact \sfdqcp and $\sfdqcp^{1,1}$ secure respectively, as follows from the following observations:
\begin{itemize}
    \item both schemes satisfy flip detection security given that the protected function class exhibits a property we call \emph{flip unlearnability}, and
    \item flip unlearnability is actually equivalent to unlearnability.
\end{itemize}

In \cref{sec:const} we use $\sfdqcp^{n,k}$ secure copy-protection schemes to obtain $\sudqccaa^{n,k}$ secure uncloneable decryptors. Combined with the oracle instantiation this implies that $\sudqccaa$ secure uncloneable decryptors exist relative to a quantum oracle, and $\sudqccaa^{1,1}$ secure uncloneable decryptors exist relative to a classical oracle. In contrast, the best security achieved by previously known construction is $\sudqcpa^{1,1}$.

\subsection{Copy Protection with Random Input Oracles}\label{ssec:qcp-ria}

In the security proofs in \cref{sec:const} it is often comfortable to modify \cref{def:wqcp} and \cref{def:fdqcp} to allow the adversary to evaluate the copy protected function at random points. In this section, we argue that this modification does not imply a stronger notion of security.

Intuitively, this holds since the pirate can sample sufficiently many random strings and use her copy protected program to evaluate these points. We now formalize this intuition.

\begin{definition}[Random Input Oracle (\textsf(RIA))]\label{def:ria}
    For any function $f$ with domain $D$ let $\sR(f)$ be the oracle which takes no input and outputs $(r,f(r))$ where $r\gets D$. 
\end{definition}

\begin{definition}\label{def:qcp_ria}
    The \swqcpriank (resp. \sfdqcpriank) \emph{game} is defined as the \swqcpnk (resp. \sfdqcpnk) game with modification that each freeloader $\sF_i$ is given access to $\sR(f)$ (where $f$ is the function sampled by $\sC$).
\end{definition}

\begin{lemma}\label{thm:qcp-ria}
    If \sqcp is \swqcpnk (resp. \sfdqcpnk) secure then it is also \swqcpriank (resp. \sfdqcpriank) secure.
\end{lemma}

\begin{proof}
    We prove the proposition for \swqcp, though the proof for \sfdqcp is identical.
    
    Let $\sA = (\sP,\sF_1,\ldots,\sF_{n+k})$ satisfy that 
    $$\E[\swqcpria^{n.k}_{\sqcp}(\sA,\secpar)]= \mu\text{,}$$
    we construct $\sA'=(\sP',\sF_1',\ldots,\sF_{n+k}')$ such that
    $$\E[\swqcp^{n,k}_{\sqcp}(\sA',\secpar)]= \mu\text{.}$$
        
    Since the freeloaders $\sF_i$ are polynomial, there is a polynomial bound $q$ on the accumulated number of queries they make to $\sR(f)$.
    
    The pirate $\sP'$ the pirate samples $r_1,\ldots,r_q$ uniformly at random and creates the list $L = ((r_i,\qcp.\qeval(\rho,r_i))_{i=1}^q$. They then simulate the pirate $\sP(\rho^{\otimes n})$ to obtain the states $\sigma_1,\ldots,\sigma_{n+k}$. Finally, she gives each freeloader $\sF'_i$ the state $L\otimes \sigma_i$.
        
    The freeloader $\sF_i'$ simulates $\sF(\sigma_i)$. Whenever $\sF_i$ queries $\sF(f)$, $\sF'_i$ responds with a previously unused pair from $L$. She resume the simulation until obtaining an output $b$ which she outputs herself.
        
    The view of $\sA$ is exactly the same in the $\swqcpria^{n,k}$ game and in the simulation above, so the output of both interactions is identically distributed. Since the output of $\sF_i'$ is simply the same as the output of $\sF_i$, it follows that $\swqcpria^{n,k}_{\cp_\bbf}(\sA,\lambda)$ and $\swqcp^{n,k}_{\cp_\bbf}(\sA',\lambda)$ distribute identically.
\end{proof}
%auto-ignore
\section{Syntax and Security of Uncloneable Decryptors}\label{sec:dcs}

An \emph{uncloneable decryptors encryption scheme} is a symmetric encryption scheme that allows the owner of the secret key to derive quantum states we call \emph{decryptors}. A decryptor could be used to decrypt messages but is unfeasible to clone, even given access to polynomially many decryptors derived from the same secret key.

The syntax of this primitive is a slight generalization of the notion of \emph{single decryptor schemes} introduced in \cite{GZ20} and further discussed in \cite{CLLZ21}.

\begin{definition}(Uncloneable Decryptors Scheme)\label{def:dcs} An \emph{uncloneable decryptors scheme} $\ud$ is comprised of the following five QPT procedures:
\begin{itemize}
    \item $\sk \leftarrow \keygen(\secparam) $,
    \item $\rho \leftarrow\decgen(\sk)$,
    \item $c \leftarrow \enc_\sk(m)$, 
    \item $m\gets\dec_\sk(c)$, and
    \item $m\leftarrow \qdec(\rho,c)$.
\end{itemize}
where $\dec$ is deterministic, $\enc$ and $\keygen$ implement classical functions.

$\dcs$ is \emph{correct} if for any $m$ it holds that

$$
    \PP\left[
        \begin{matrix}
            \sk\leftarrow\keygen(\secparam) \\
            \rho \leftarrow \decgen(\sk) \\
            c \gets \enc(\sk,m) \\
        \end{matrix}
        :
        \qdec(\rho,c) = \dec_\sk(c) = m
    \right]=  1\text{.}
$$
\end{definition}

Our definition slightly differs from \cite{GZ20}'s as it has an explicit procedure for producing decryptors. It is also different from \cite{CLLZ21}'s definition in two ways: it explicates the classical decryption circuit (which uses the secret key rather than the quantum decryptor) and assumes that the underlying encryption scheme is symmetric (whereas \cite{CLLZ21} assume the underlying encryption is asymmetric).

In \cref{ssec:dcs_sec} we introduce notions of security suitable for uncloneable decryptors, which generalize the notions of security introduced in previous works. By removing the $\decgen$ and $\qdec$ procedures from the syntax of an uncloneable decryption scheme, one obtains the \emph{underlying symmetric encryption scheme.} In \cref{ssec:dcs_sec} we formally introduce the underlying scheme and show that the security of the uncloneable decryptors scheme implies the security of the underlying encryption scheme. In \cref{ssec:dcs_uncond} we show that even the weakest form of security is unattainable against an unbounded adversary with access to arbitrary polynomially many decryptors.

\subsection{Security Notions}\label{ssec:dcs_sec}

Here we provide the notions of \sudqcpa, \sudqcca and \sudqccaa security for uncloneable decryptors. These are adaptations of the respective notions for symmetric encryption schemes (recall \cref{def:ind-qccax}).

These definitions combine the security notions of quantum copy protection (see \cref{def:wqcp}) and symmetric encryption (see \cref{def:ind-qccax}). We retain the form of the security game of quantum copy protection, where a pirate $\sP$ receives $n$ copies of a program from the challenger $\sC$ and creates $n+k$ quantum states. These states are then given to \emph{distinguishers} $\sD_1,\ldots,\sD_{n+k}$ (which replace the freeloaders $\sF_1,\ldots,\sF_{n+k})$). The challenger then plays against each distinguisher a game similar to the indistinguishability game for symmetric encryption schemes. Namely, $\sD_i$ needs to distinguish between ciphertexts of two plaintexts of her choosing. Like in symmetric encryption, we make our security notions progressively stronger by affording the adversary more forms of oracle access.

\begin{definition}[\sud security]\label{def:dcs-ind}
    Let $\dcs$ be an uncloneable decryptors scheme, and let $\sA=(\sP,\sD_1,\ldots,\sD_{n+k})$ be procedures, the $\sudnk_\dcs(\sA,\lambda)$ game is defined as follows:
    \begin{enumerate}
        \item $\sC$ generates $\sk\leftarrow\keygen(1^\secpar)$ and $n$ decryptors $\rho_1\leftarrow\decgen(\sk),\ldots,\rho_n\leftarrow \decgen(\sk)$ and samples $b_1,\ldots,b_{n+k}\gets\zo$,
        \item $\sP(\rho_1,\ldots,\rho_n)$ creates states\label{step:dcs_pirate} $\sigma_1,\ldots,\sigma_{n+k}$, and $n+k$ pairs of plaintexts $(m_0^i,m_1^i)$ with $|m_0^i|=|m_1^i|$,
        \item For $i=1,\ldots,n+k$, $\sC$ calculates $c_i \gets \enc_\sk(m^i_{b_i})$,
        \item $\sD_i(\sigma_i,b_i)$ outputs a bit $\beta_i$,
        \item the output of the game is the number of indices $i$ for which $b_i=\beta_i$.
    \end{enumerate}
    
    We say that $\dcs$ is \emph{\sudnk secure} if for any QPT $\sA=(\sP,\sD_1,\ldots,\sD_{n+k})$ it holds that
    $$\E \left[\sindnk_{\dcs}(\sA,\lambda) \right] \le n + k/2 + \negl\text{.}$$
    
    We say that $\dcs$ is \emph{\sud secure} if it is \sudnk secure for any $n,k\in \poly$.
\end{definition}

This notion is extended by augmenting the freeloaders with oracle access. For convenience, we marked the oracles added to \cref{def:dcs-indx} with a red underline. Removing all underlined expressions exactly recovers \cref{def:dcs-ind}.

\begin{definition}[\sudqx security]\label{def:dcs-indx}
    Let $\dcs$ be an uncloneable decryptors scheme, and let $\sA=(\sP,\sD_1,\ldots,\sD_{n+k})$ be procedures, for any two oracles $\sO_1,\sO_2$ define the $\sudqxnk_{\sO_1,\sO_2}(\sA,\secpar)$ game:
        \begin{enumerate}
        \item $\sC$ generates $\sk\leftarrow\keygen(1^\secpar)$ and $n$ decryptors $\rho_1\leftarrow\decgen(\sk),\ldots,\rho_n\leftarrow \decgen(\sk)$ and samples $b_1,\ldots,b_{n+k}\gets\zo$,
        \item $\sP^{\redul{\ket{\enc_\sk},\ket{\sO_1}}}(\rho_1,\ldots,\rho_n)$ creates states $\sigma_1,\ldots,\sigma_{n+k}$, and $n+k$ pairs of plaintexts $(m_0^i,m_1^i)$ with $|m_0^i|=|m_1^i|$,
        \item For $i=1,\ldots,n+k$, $\sC$ calculates $c_i \gets \enc_\sk(m^i_{b_i})$,\label{step:dcs_chal}
        \item $\sD_i^{\redul{\ket{\enc_\sk},\ket{\sO_2}}}(\sigma_i,c_i)$ outputs a bit $\beta_i$,\label{ent:dcs-cca1vs2}
        \item the output of the game is the number of indices $i$ for which $b_i=\beta_i$.
    \end{enumerate}
    
    By explicating the oracles $\sO_1,\sO_2$ we define the games which will define our security notions. Let $\bot$ designate a trivial oracle which always responds with $\bot$.
    \begin{itemize}
        \item $\sudqcpank = \sudqxnk_{\bot,\bot}$,
        \item $\sudqccank = \sudqxnk_{\dec_\sk,\bot}$,
        \item $\sudqccaank = \sudqxnk_{\dec_\sk,\dec_\sk\setminus c_i}$ where $c_i$ is the output of the challenge query given to $\sD_i$ in step~\ref{step:dcs_chal}.
    \end{itemize}
    
    We say that $\dcs$ is \emph{\sudqcpank secure} if for any QPT $\sA=(\sP,\sD_1,\ldots,\sD_{n+k})$ it holds that
    $$\E \left[\sudqcpank_{\dcs}(\sA,\lambda) \right] \le n + k/2 + \negl\text{.}$$ We say that $\dcs$ is \emph{\sudqcpa secure} if it is \sudqcpank secure for any $n,k\in \poly$.
    
    \emph{\sudqccank}, \emph{\sudqcca}, \emph{\sudqccaank} and \emph{\sudqccaa security} are defined similarly.
\end{definition}

Our constructions below are formed by first creating a scheme that only supports encrypting a single bit (namely uncloneable \emph{bit} decryptors) and then extending them to messages of unrestricted (polynomial) length. The security game takes on a simpler form in the bit encryption setting, as there are only two possible ciphers.

\begin{definition}[\sobudqx security]\label{def:dcs-obindx}

    The \sobudqxnk games are defined almost exactly like the \sudqxnk of  \cref{def:dcs-ind} and \cref{def:dcs-indx}, but with the following modifications:
    \begin{itemize}
        \item in step~\ref{step:dcs_pirate}, $\sP$ does not create any plaintexts, and
        \item in step~\ref{step:dcs_chal}, $c_i \gets \enc_\sk(b_i)$.
    \end{itemize}

    We say that $\dcs$ is \emph{\sobudqx secure} if for any QPT $\sA=(\sP,\sD_1,\ldots,\sD_{n+k})$ it holds that
    $$\E \left[\sobudqxnk_{\dcs}(\sA,\lambda) \right] \le n + k/2 + \negl\text{.}$$ We say that $\dcs$ is \emph{\sobudqx secure} if it is \sobudqxnk secure for any $n,k\in \poly$.
\end{definition}

\subsection{The Underlying Symmetric Encryption Scheme}\label{ssec:dcs_assoc}

An uncloneable decryptors scheme is a symmetric encryption scheme with added functionality. When removing the extra functions, we remain with a run-of-the-mill encryption scheme which we call the \emph{underlying scheme}.

\begin{definition}[Underlying Scheme]\label{def:assoc-se} Let $\dcs$ be an uncloneable decryptors encryption scheme, the \emph{underlying encryption scheme} is $$\se_\dcs = (\dcs.\keygen,\dcs.\enc,\dcs.\dec)\text{.}$$ 
\end{definition}

Notably, $\sudqx^{1,k}$ security for any $k$ implies that the underlying scheme admits $\sindqx$ security.

\begin{proposition}\label{thm:DCS-IND}
    If $\dcs$ is $\sud^{1,k}$ (resp. $\sudqcpa^{1,k}$, $\sudqcca^{1,k}$ and $\sudqccaa^{1,k}$) secure (see \cref{def:dcs-ind,def:dcs-indx}) for any $k\ge 1$ then $\se_\dcs$ is $\sind$ (resp. \sindqcpa, $\sindqcca$ and $\sindqccaa$) secure (see \cref{def:ind-qccax}).
\end{proposition}

\begin{proof}
    We prove the result for $k=1$; our argument generalizes straightforwardly to any $k$.

    The proof idea is straightforward: we need to win two distinguishing games, we win one with certainty using the decryptor, and we use the \sindqx adversary for the second one. If the \sindqx wins with probability $\frac{1}{2}+\epsilon$ it follows that the expected number of correct distinguishers is $1+\frac{1}{2}+\epsilon$, and it follows from the \sudqx security that $\epsilon = \negl$.
    
    Let $\sA = (\sA_1,\sA_2)$ be a QPT adversary for the $\sindqx_{\se_\dcs}$ game, we describe an adversary $(\sP,\sD_1,\sD_2)$ for the $\sudqx^{1,1}_{\dcs}$ game.
    
    We note that for any security notion \sindqx, the corresponding notion \sudqx satisfies that $\sP$ has exactly the same oracle access as $\sA_1$ and $\sD_2$ has exactly the same oracle access as $\sA_2$.
    
    After obtaining the decryptor $\rho$ from $\sC$, $\sP$ simulates $\sA_1$ her own oracle access to answer oracle calls, until obtaining the plaintexts $m_0,m_1$ and auxiliary data $\sigma$. She gives $\rho$ to $\sD_1$ and $\sigma$ to $\sD_2$. She then gives the adversary the following pairs of messages $(m_0^1=0,m_1^1=1),(m_0^2=m_0,m_1^2=m_1)$ (that is, $\sD_1$ needs to distinguish the ciphers of $0$ and $1$ while $\sD_2$ needs to distinguish the ciphers of $m_0$ and $m_1$).
    
    The distinguisher $\sD_1$ outputs $\ud.\qdec(\rho,c_1)$. The distinguisher $\sD_2$ simulates $\sA_2(\sigma)$, using her own oracle access to answer any queries, and outputs the result.
    
    Let $\beta_1,\beta_2$ be the bits sampled by $\sC$ in the challenge query, and $b_1,b_2$ be the outputs of $\sD_1,\sD_2$ respectively, then $$ \prob{b_1 = \beta_1} + \prob{b_2 = \beta_2} = \E \left[\sudqx^{1,1}_{\dcs}(\sA',\secpar) \right] \le 1 + 1/2 + \negl$$ where the inequality is due to the $\sudqx^{1,1}$ security of $\dcs$.
    
    From the correctness of $\dcs$ we have that $\prob{b_1=\beta_1} = 1$. From the construction of $\sD_2$, we have that $\prob{b_2=\beta_2} = \prob{\sindqx_{\se_\dcs}(\sA,\lambda) = 1}$. Plugging these probabilities into the inequality above we obtain (after some rearrangement) that
    
    $$\prob{\sindqx_{\se_\dcs}(\sA,\lambda) = 1} \le 1/2 + \negl$$
    as needed.
\end{proof}

\begin{remark}
    The same argument can be used almost verbatim to show for any $k\ge 1$ that $\sobudqx^{1,k}$ security implies the corresponding $\sindqx$ security for \emph{bit encryption} schemes.
\end{remark}

\subsection{Impossibility of Unconditional Security}\label{ssec:dcs_uncond}

In this section, we discuss the impossibility of unconditional security against an adversary that can request polynomially many decryptors.

The authors of \cite{GZ20} consider a scenario where the adversary is given access to polynomially many \emph{ciphers} of random plaintexts (without even being given the plaintexts themselves). They prove that unconditional security is not obtainable against such adversaries. This implies, in particular, that no scheme is unconditionally $\sudqcpa^{1,k}$ secure for any $k$.

In this section, we prove that if we allow an arbitrary polynomial number of decryptors, $\sud$ security is also impossible. That is, for large enough polynomial $n$ and any polynomial $k$ there exists an unbounded $\sindnk$ adversary with a non-negligible advantage.

Our result is incomparable with the impossibility result of \cite{GZ20}: the adversary we consider is stronger in the sense that she has access to many decryptors but is weaker in the sense that she does not have access to any ciphers.

The proof is a straightforward application of a technique called \emph{shadow tomography}, first considered in \cite{Aar18}. Consider a family $\sE$ of two-outcome measurements acting on $D$ dimensional states. Let $\varepsilon>0$ be some error tolerance. Say you have access to an unrestricted number of copies of some state $\rho$. The \emph{shadow tomography task} is to compute for each $E\in \sE$ an estimation $s_E$ such that $$\comma {\forall E\in \sE}{|s_E - \tr(E\rho)|} \le \varepsilon{.}$$ The following theorem bounds the number of copies of $\rho$ required to achieve this task.

\begin{theorem}[\cite{Aar18}, Theorem~2]\label{thm:shadow}
    The shadow tomography task could be solved with success probability $1-\delta$ using $$\tilde{O}\left(\log\left(\frac{1}{\delta}\right)\cdot \log ^4 |\sE| \cdot \log D \cdot \varepsilon^{-4}\right)$$ copies of $\rho$.
\end{theorem}

\cref{thm:shadow} is used in \cite[Theorem~7]{Aar18} to prove the impossibility of unconditionally secure quantum money. Our proof is an adaptation of their argument to uncloneable decryptors.

\begin{theorem}\label{thm:impos}
    Let $\dcs$ be an uncloneable decryptors scheme, there exists $n=\poly$ such that for any $k$ there exists a (computationally unbounded) adversary $\sA$ such that $$\mathbb{E}\left[\sudnk_\dcs (\sA,\lambda) \right] \ge n + k - \negl\text{.}$$
\end{theorem}

\begin{proof}
    Let $\ell$ be the length of a ciphertext for a plaintext of length $1$. For each string $c\in \{0,1\}^\ell$ let $E_c$ be a two outcome measurement which, on input $\rho$, measures $\qdec(\rho,c)$ in the computational basis and accepts if and only if the result is $0$. Let $\sE = \{E_c\}$.
    
    Let $D$ be the dimension of the decryptors $\rho$ produced by $\decgen(\sk)$. Note that since all procedures of $\dcs$ are QPT, it follows that $\rho$ is a state on $\poly$ many qubits. That is, $D=2^{\poly}$. It also follows that $|c|=\poly$ whereby $|E|$ is exponential in $\lambda$.
    
    It follows from \cref{thm:shadow} that there exists $n = \poly$ such that given $\rho^{\otimes n}$ it is possible to calculate estimates $s_{E_c}$ such that with probability $1- 2^{-\secpar} = 1 - \negl$, $$\comma{\forall c\in \{0,1\}^{p(\secpar + |m_0|)}}{|s_{E_c} - \prob{\qdec(\rho,c)=m_0}| \le \frac{1}{4}}\text{.}$$
    
    The pirate $\sP$ performs shadow tomography on the set $\sE$ of circuits defined above, using the state $\rho^{\otimes n}$ obtained from the challenger to obtain estimations $s_{E_c}$, which they transmit to $\sD_1,\ldots,\sD_{n+k}$. For every $i$ they transmit to $\sC$ the pair $(m_0,m_1)$.

    The distinguishers $\sD_i$ return $0$ if and only if $s_{E_{c_i}} > \frac{1}{2}$.

    The correctness of $\dcs$ implies that if $\rho \gets \decgen(\sk)$ where $\sk\gets \keygen(\secparam)$ then $c\gets \enc_\sk(m_0)$ implies that $\prob{E(\rho) = 0} = 1$, whereby $S_{E_c} > \frac{3}{4}$ with probability $1-\negl$. Similarly that if $c\gets \enc_\sk(m_1)$ then $\prob{E(\rho) = 1} = 1$ whereby $S_{E_c} < \frac{1}{4}$ with probability $1-\negl$.
    
    It follows that $$\mathbb{E}\left[\sudnk_\dcs (\sA,\lambda) \right] = (1-\negl)(n+k) = n+k-\negl$$ as needed.
\end{proof}

\begin{remark}
    Note that the proof above uses the plaintext $m_b=b$. This implies that \cref{thm:impos} holds also for \sobud security.
\end{remark}

\subsection{Extendability}\label{ssec:def_extend}

In order to construct uncloneable decryptors, we first construct uncloneable bit decryptors and then extend them. One way to do so is by the following "bit-by-bit" transformation.

\begin{definition}[Extended Scheme]\label{def:extend}
    Let \ud be an uncloneable decryptors scheme which supports messages of length $1$, define the \emph{extension} of \ud to be the following scheme \udext:
    \begin{itemize}
        \item $\udext.\keygen \equiv \ud.\keygen$,
        \item $\udext.\decgen \equiv \ud.\decgen$,
        \item $\udext.\enc_\sk(m)$ outputs $(c_1,\ldots,c_{|m|})$ where $c_i\gets \ud.\enc_\sk(m_i)$,
        \item $\udext.\dec_\sk((c_1,\ldots,c_\ell)$ outputs $m_1\ldots m_\ell$ where $m_i\gets \ud.\dec_\sk(c_i)$,
        \item $\udext.\qdec(\rho,(c_1,\ldots,c_\ell))$ outputs $m_1\ldots m_\ell$ where $m_i\gets \ud.\qdec(\rho,c_i)$.
    \end{itemize}
\end{definition}

\begin{definition}[\sudqxnk extendability]
    Let \ud be an uncloneable decryptor encryption scheme, we say that \ud is \emph{\sudqxnk extendable} if \udext is \sudqxnk secure (see \cref{def:dcs-ind} and \cref{def:dcs-indx}).
    
    We say that \ud is \emph{\sudqx extendable} if it is \emph{\sudqxnk extendable} for any $n,k=\poly$.
\end{definition}

It is trivial to check that \sudqxnk extendability implies \sobudqxnk security. Unfortunately, the converse is not generally true. Indeed, we will see in \cref{con:dcs_cca2_res} that given a \swqcp secure copy protection scheme for any \bbf, one can construct an uncloneable decryptors scheme which is \sobudqccaa secure but which is not even $\sud^{1,1}$ extendable. In other words, the transformation described by \cref{def:extend} does not afford a generic method to extend a length restricted scheme to an unrestricted scheme. Nevertheless, the notion of extendability will be helpful in the following constructions.
%auto-ignore
\section{Constructions}\label{sec:const}

Having established the relevant definitions and security notions in \cref{sec:dcs} we now turn to present several constructions and transformations of construction that strengthen their security.

\subsection{\textsf{UD1-qCPA} Security}\label{ssec:con_ud1_qcpa}

We first explain how to obtain \sobudqcpank security from a \swqcpnk secure copy protection scheme for any \bbf.

This construction is inspired by the standard construction of symmetric bit encryption from pseudo-random functions (see e.g.  \cite[Construction~5.3.9]{Gol04}). In this construction, the key to the $\prf$ is used as a key to an encryption scheme, and a bit $b$ is encrypted by sampling a uniformly random string $r$ and outputting the pair $(r,\prf_\sk(r)\oplus b)$. We follow a similar approach, but replace the $\prf$ with a \swqcp secure copy protection for a \bbf.

\begin{construction}[\sobudqcca Secure Scheme from \swqcp secure \bbf]\label{con:dcs_cca1_res} Let \bbf be a balanced binary function (see \cref{def:bbf}) with input length $\ell$. Let $\cp_\bbf$ be a copy protection scheme for \bbf. Define the scheme $\dcsobcca$:
\begin{itemize}
    \item $\dcsobcca.\keygen \equiv \bbf.\samp$.
    \item $\dcsobcca.\chipgen \equiv \cp.\ptect$.
    \item $\dcsobcca.\enc_\sk(b) \to (r,b\oplus \bbf.\eval_\sk(r))$ where  $r\gets\zo^{\ell}$.
    \item $\dcsobcca.\dec_\sk((r,\hat{b}))\to \hat{b}\oplus \bbf.\eval_\sk(r)$.
    \item $\dcsobcca.\qdec(\rho,(r,\hat{b})) \to \hat{b}\oplus \cp_\bbf.\qeval(\rho,r)$.
\end{itemize}
\end{construction}

The correctness of \cref{con:dcs_cca1_res} follows from the correctness of $\bbf$ and $\cp_\bbf$.

\begin{remark}
    Note that this construction is manifestly not \sudqccaa secure: given the challenge cipher $(r, m\oplus f_k(r))$, a distinguisher could make a decryption query on $(r, 1\oplus m\oplus f_k(r))$ to obtain $1\oplus m$.
\end{remark}

\begin{proposition}\label{thm:con_1bit_dcs_cca1}
    If $\qcp_\bbf$ is \swqcpnk secure (see \cref{def:wqcp}) then the scheme \dcsobcca (see \cref{con:dcs_cca1_res}) is \sobudqcpank secure (see \cref{def:dcs-obindx}).
\end{proposition}

\begin{proof}

    By \cref{thm:qcp-ria}, $\qcp_\bbf$ is \swqcpriank secure.

    Assume the QPT adversary $\sA=(\sP,\sD_1,\ldots,\sD_{n+k})$ satisfies that $$\E[\sobudqccank_{\dcsobcca}(\sA,\secpar)]=\mu\text{,}$$ we construct QPT procedures $\sA'=(\sP',\sF_1,\ldots,\sF_{n+k})$ for which $$\E[\swqcpria^{n,k}_{\cp_\bbf}(\sA',\secpar)]=\mu-\negl .$$ 
   
    That $\cp_\bbf$ is \swqcpria secure implies that $$\E[\swqcpria^{n,k}_{\cp_\bbf}(\sA',\secpar)] \le n + k/2 + \negl$$ whereby the equality above would imply that $\mu \le n + k/2 + \negl$.
        
    Upon getting $\rho_1,\ldots,\rho_n$, the pirate $\sP'$ simulates $\sP(\rho_1,\ldots,\rho_n)$. When $\sP$ queries $\enc_\sk$, $\sP'$ queries $\sR(\bbf.\eval_k)$ to obtain a pair $(r,f_\sk(r))$ and applies the unitary $\ket{m,x}\mapsto\ket{m,x\oplus(r,f(r)\oplus m)}$ to the input of the first query. When the simulation is concluded, $\sP'$ obtains $\sigma_1,\ldots,\sigma_{n+k}$ which she gives to the freeloaders $\sF_1,\ldots,\sF_{n+k}$.
        
    Once the freeloader $\sF_i$ is given $\sigma_i$ from the pirate and $x_i$ from the challenger, she samples a random bit $\beta_i$ and invokes $\sD_i(\sigma_i)$ with the cipher $(x,\beta_i)$ to obtain output $b_i$. She returns $y_i= b_i\oplus \beta_i$.
        
    Note that $(x,\beta_i)$ is a valid cipher for $\beta_i \oplus \bbf.\eval_\sk(x)$. Hence, if the output of $\sF_i$ is correct then $b_i=\beta_i\oplus \bbf.\eval_\sk(x)$ whereby $y_i = b_i\oplus \beta_i = \bbf.\eval_\sk(x)$ as needed. Thus, the number of distinguishers who answer correctly is exactly the number of freeloaders who answer correctly.
        
    The balancedness of $\bbf$ implies that if $x$ is uniformly random then $m\oplus \bbf.\eval_\sk(x)$ is $0$ with probability $1/2 + \negl$. It follows that the statistical difference between the view of $\sF_i$ when simulated by $\sD_i$ and the view of $\sD_i$ in the $\sobudqcpank_{\dcsobcca}(\sA,\lambda)$ game is negligibly close, whereby the expected number of distinguishers who answered correctly is negligibly close to $\mu$.
\end{proof}

\subsubsection{Implications}\label{sssec:impl}

Theorems \cref{thm:con_1bit_dcs_cca1} and \cref{thm:DCS-IND} together imply that the existence of a $\swqcp^{1,1}$ copy protectable \bbf implies the existence of a post-quantum \sindcpa secure bit encryption scheme. Such schemes can are known to imply the existence of post-quantum one way functions. This line of argument boils to:

\begin{corollary}\label{cor:owf}
    The existence of a $\swqcp^{1,1}$ secure copy protection scheme for a \bbf implies the existence of post-quantum one-way functions.
\end{corollary}

Consider the bit encryption scheme described at the top of  \ref{ssec:con_ud1_qcpa}. Examination of the analysis of its \sindcpa security (e.g. \cite[Proposition~5.4.12]{Gol04}) reveals that for the construction to be secure, it is sufficient that the construction remains secure if we only require that the underlying function is a \emph{weak} \prf. That is, it is infeasible to distinguish it from a truly random function for an adversary with access to its value on polynomially many \emph{uniformly random} points. Furthermore, it is possible to show that being a weak $\prf$ is also a necessary condition. This leads to the following corollary, the proof thereof is deferred to \cref{app:ggm}:

\begin{corollary}\label{cor:prf}
    If there exists a $\swqcp^{1,1}$ secure copy protection scheme for a binary balanced function \bbf, then \bbf is a weak \prf.
\end{corollary}

\subsection{\textsf{UD1-qCCA1} Security}\label{sssec:obindcca}

The construction \cite[Construction~5.3.9]{Gol04} which we discussed before is actually \sindcca secure \cite[Proposition~5.4.18]{Gol04}. Intuitively, the proof of this claim follows by noting that the value of the function at the point used to mask the plaintext seems independent of the value of the function on all points which were required to answer previous encryption queries to computationally bounded adversaries, as mandated by the \prf property. Hence, having access to decryption oracle \emph{before} seeing the challenge ciphertext does not benefit the adversary.

In trying to carry this idea to the context of uncloneable decryptors, one runs into a difficulty: Our ability to answer decryption queries relies on our ability to evaluate the underlying \bbf, which we are only able to do with the help of the copy-protected programs given by the challenger. However, simulating the pirate might modify the copy-protected programs in a way that makes them unusable.

It is tempting to try to sidestep this by means of rewinding: every time the pirate makes a query, apply it in reverse to recover the copy-protected programs, use them to respond to the query, and apply the pirate forward to get back to the querying point. The problem is that the inputs to decryption queries might depend on measurement outcomes (or equivalently, if we simulate the pirate coherently, the queries might become entangled with the auxiliary qubits to which we store the measurement outcomes).

This issue could be completely circumvented if we allow the \swqcp adversary to have one additional program they could use to respond to queries. Following this line of argument, one can prove:

\begin{proposition}\label{thm:con_1bit_dcs_cca1_really}
    If $\qcp_\bbf$ is $\swqcp^{n+1,k}$ secure (see \cref{def:wqcp}) then the scheme \dcsobcca (see \cref{con:dcs_cca1_res}) is \sobudqccank secure (see \cref{def:dcs-obindx}).
\end{proposition}

We do not provide a formal proof as (assuming secure digital signatures) this result is superseded by the construction of the next section.

\subsection{\textsf{UD1-qCCA2} Security}\label{sssec:obindccaa}

Here we employ digital signatures to generically transform \sobudqcpank secure decryptors to \sobudqccaank secure decryptors.

The transformation is conceptually similar to the standard transformation of \sindcpa secure symmetric encryption schemes into \sindccaa secure symmetric encryption schemes by means of \emph{message authentication codes} (see e.g. \cite[Proposition~5.4.20]{Gol04}).

Informally, Message authentication codes are a way to produce \emph{tags} for strings such that anyone holding a secret key can tag messages as well as verify that other messages have been tagged with the same key, but such that it is infeasible to create valid tags without knowing the secret key. By modifying the scheme to tag ciphers at encryption, and only decrypting properly tagged messages (returning $\bot$ otherwise), we render the decryption oracles useless, since creating valid ciphers other than ones obtained from encryption queries becomes unfeasible.

The \emph{encrypt-than-MAC} paradigm described above is unsuitable for our needs since the adversary has access to a decryptor. Since the decryptor should allow decrypting messages (and in particular, verifying tags), it should somehow contain the authentication key. However, simply affording this key in the clear would allow the adversary to tag messages themselves.

We circumvent this by using a digital signature (see \cref{ssec:ds}) rather than a message authentication code, which allows us to separate tagging (henceforth called \emph{signing}) from verification.

In order to answer encryption queries, we need to record decryption queries and their results. This is impossible to do for queries in superposition in general but becomes possible when we assume that the scheme is \emph{decoupled}. That is, we assume that the encrypting $0$ and $1$ using the \emph{same randomness} appears independent (at least to a computational adversary). Fortunately, it is easy to transform any scheme to a decoupled scheme while retaining its security, which allows us to assume without loss that the scheme we wish to transform is already decoupled. In \cref{app:decouple} we formally define and show how to decouple uncloneable bit decryptors (which is sufficient for our applications) and sketch a decoupling procedure for uncloneable decryptors in general.

\begin{construction}[\sobudqccaa Uncloneable Decryptors from \sobudqcpa Uncloneable Decryptors and \sseufcma Signatures]\label{con:dcs_cca2_res}
    Let $\dcsdec$ be a decoupled (see \cref{def:decoupled}) uncloneable decryptors scheme, and let $\ds$ be a deterministic digital signature scheme (see \cref{def:ds}). Define the scheme \dcsobccaa as following:
    \begin{itemize}
        \item $\dcsobccaa.\keygen(\secparam)$ outputs $(\sk_\dcs,\sk_\ds,\pk_\ds)$ where:
        \begin{itemize}
            \item $\sk_\dcs\gets \dcsdec.\keygen(\secparam)$, and 
            \item $(\sk_\ds,\pk_\ds)\gets \ds.\keygen(\secparam)$.
        \end{itemize}
        \item $\dcsobccaa.\chipgen_{(\sk_\dcs,\sk_\ds,\pk_\ds)} $ outputs $ \rho\otimes \pk_\ds$ where $\rho\gets \dcsdec.\chipgen(\sk_\dcs)$.
        \item $\dcsobccaa.\enc_{(\sk_\dcs,\sk_\ds,\pk_\ds)}(b) $ outputs $ (c,s)$ where $c\gets \dcsdec.\enc_{\sk_\dcs}(b)$ and $s\gets \ds.\sign_{\sk_\ds}(c)$.
        \item $\dcsobccaa.\dec_{(\sk_\dcs,\sk_\ds,\pk_\ds)}((c,s))$ outputs  $$\begin{cases} \dcsdec.\dec_{\sk_\dcs}(c) & \ds.\ver_{\pk_\ds}(c,s)=1 \\ \bot & \text{else} \end{cases}\text{,}$$
        \item $\dcsobccaa.\qdec(\rho\otimes \pk_\ds,(c,s))$ outputs  $$\begin{cases} \dcsdec.\qdec(\rho,c) & \ds.\ver_{\pk_\ds}(c,s)=1 \\ \bot & \mbox{else} \end{cases}\text{.}$$
    \end{itemize}
\end{construction}

The correctness of $\dcsobccaa$ of \cref{con:dcs_cca2_res} follows from the correctness of $\dcsdec$ and $\ds$.

\begin{proposition}\label{thm:dcs_cca2_res}
    If $\dcsdec$ is \sobudqcpank secure (see \cref{def:dcs-indx}) and $\ds$ is  \sseufcma secure (see \cref{def:euf-qcma}) then \dcsobccaa from \cref{con:dcs_cca2_res} is \sobudqccaank secure (see \cref{def:dcs-indx}).
\end{proposition}

\begin{proof}
    Let $\sA = (\sP,\sD_1,\ldots,\sD_{n+k})$ be an adversary for the $\sudqccaank_{\dcsobccaa}$ game. We construct an adversary $\sA' = (\sP',\sD_1',\ldots,\sD_{n+k}')$ for the $\sudqcpank_{\dcsdec}$  game such that
    \begin{equation}\label{eq:bla}
        \E\left[\left|\sudqccaank_{\dcsobccaa}(\sA,\lambda)-\sudqcpank_{\dcsdec}(\sA',\lambda)\right|\right]< \negl\text{,}
    \end{equation}
    the \sobudqccaank security of $\dcsobccaa$ will then follow from the \sobudqcpank security of $\dcsdec$.
    
    In the following, we describe how $\sA'$ simulates answers to decryption oracle calls made by $\sA$. We often refer to the encryption oracle as the \emph{actual} oracle and to the responses made by $\sA'$ to such calls as the \emph{simulated} oracle.
    
    After being given $\rho^{\otimes n}$ from $\sC$, the pirate $\sP'$:
    \begin{itemize}
        \item Generates $\sk_\ds,\pk_\ds\gets \ds.\keygen(\secparam)$.
        \item Simulates $\sP((\rho\otimes\pk_\ds)^{\otimes n})$, responding to encryption queries by:
        \begin{itemize}
            \item using her encryption oracle on inputs $0$ and $1$ to obtain $c_0$ and $c_1$,
            \item computing $s_b \gets \ds.\sign_{\sk_\ds}(c_b)$,
            \item storing $(s_b,c_b,b)$ to a list $L_\sP$, and
            \item applying the unitary $\ket{b,x}\mapsto \ket{b,x\oplus(s_b,c_b)}$,
        \end{itemize}
        until obtaining the states $\sigma_1,\ldots,\sigma_{n+k}$. Gives each $\sD'_i$ the state $\sigma_i\otimes \pk_\ds \otimes \sk_\ds \otimes L_\sP$.
    \end{itemize}

    Each distinguisher $\sD'_i$ simulates $\sD_i(\sigma_i\otimes \pk_\ds)$ on the cipher $c_i$ given to her by $\sC$. $\sD'_i$ answers encryption calls the same way $\sP'$ did, storing signature-cipher-plaintext triplets into a list $L_i$ which we assume is initially a copy of $L_\sP$. She answers decryption calls by applying the unitary $$\ket{(c,s),x}\mapsto \ket{(c,s),x\oplus f(c,s)}$$ where
    $$f(c,s) = 
    \begin{cases} 
        b & (c,s,b)\in L_i\\
        \bot  & \mbox{otherwise}
    \end{cases}\text{.}$$ Note that $f$ is well defined since it is impossible that $c$ is a cipher of both $0$ and $1$. When the simulation ends, $\sD'_i$ outputs the output of $\sD_i$.
    
    We argue that the statistical difference between the views of $\sA$ in the simulation described above and the original $\sudqccaa_{dcs^2}$ game is $\negl$. 
    
    Intuitively, any statistical difference between the actual game and the simulation must result from differences in the outputs decryption queries (since the actual and simulated encryption oracles are identical). Before the first decryption oracle call, everything distributes identically. If all queries made by $\sA$ to encryption oracles satisfy that the response of the simulated oracle is negligibly close to the response of the actual oracle, then the state at the end of the simulation is negligibly close to the state of the actual game. In other words, if the actual and simulated views of $\sA$ by the end of the game are not statistically close, then at some point $\sA$ must have made an encryption query on an input whose result on the actual oracle is significantly different than on the simulated oracle. This can only happen if the input the the encryption query is significantly supported on pairs of the form $(c,s)$ such that $\dcsdec.\enc_{\sk_\dcs}(c)=b$ and $\ds.\ver_{\pk_\ds}(c,s)=1$ yet $(c,s,b)\notin L_i$. We use this fact to extract a signed document that was not a result of a signature query, whereby voiding the \sseufcma security of \ds and arriving at a contradiction.
    
    More precisely, for any state $\rho$ which is of the dimension of an input to the decryption oracle, let $S_j(\rho)$ be the probability that, upon measuring the first register of $\rho$, the outcome would be of the form $(c,s)$ where $(c,s,*)\notin L_i$ at the time the $j$th query was made, yet $\ds.\ver_{\pk_\ds}((r,b),s) = 1$. Let $S_j$ be the expected value of $S_j(\rho)$ where $\rho_j$ is the input to the $j$th decryption query made by $\sA'$, and let $S=\max\{S_1,\ldots,S_q\}$ where $q$ is the maximal number of decryption queries.
    
    If $S < \negl$ then the output of any decryption call made by $\sA$ in the simulation is negligibly close to the output expected by an actual decryption oracle. Since the output of $\sD_i'$ is exactly the output of $\sD_i$ it follows that $$\E\left[\left|\sobudqccaank_{\dcsobccaa}(\sA',\lambda)-\sobudqcpank_{\dcsdec}(\sA,\lambda)\right|\right]< \negl\text{.}$$
    
    We show that $S < \negl$ by constructing an adversary $\sB$ to the $\sseufcma_\ds$ game (recall \cref{def:euf-qcma}) such that  $$\prob{\sseufcma_\ds(\sB,\secpar)=1} \ge \frac{S}{ \mathsf{poly}(\secpar)}\text{,}$$ whereby it will follow from the \sseufcma security of $\ds$ that $S < \negl$.
    
    Recall that the adversary $\sB$ has access to the oracle $\sign_{\sk_\ds}$ as well as to the public key $\pk_\ds$.
    
    The adversary $\sB$ first chooses two random positive integers $j\le n+k$ and $u\le \max\{Q_i\}_{i=1,\ldots,n+k}$ where $Q_i$ is the maximal number of decryption queries to be made by $\sD_i$. $\sB$ then simulates the game $\sobudqccaank_{\dcsobccaa}(\sA',\lambda)$ (with the adversary $\sA'$ defined above) with the following modifications:
    \begin{itemize}
        \item $\sP'$ does not generate $\sk_\ds,\pk_\ds$, but is rather given $\pk_\ds$ from $\sB$ (and does not know $\sk_\ds$).
        \item All invocations of $\ds.\sign_{\sk_\ds}$ are replaced with oracle calls.
    \end{itemize}
    $\sB$ runs the simulation, responding to signature calls made by $\sA'$ by querying the signing oracle and recording the message-signature pair until $\sD'_j$ makes their $u$th query. Instead of answering the query, $\sB$ measures the first register of the input state and transmits to $\sC$ the input-output list of all queries it has made, and the output of the last measurement (if the simulation of $\sD'_j$ finishes before $u$ queries are made, $\sB$ concedes the game and the outcome is $0$).
    
    By hypothesis, at least one of the queries made by the distinguishers doing the simulation has the property that measuring the first register will result with probability $S$ with a valid signed message which was not queried by $\sB$. The probability that $\sB$ measured such a state is at least $\frac{1}{ju}$. Recall that $j \le n+k = \mathsf{poly}(\secpar)$, and that $q$ is bounded by the number of oracle queries made by a QPT procedure, whereby $q \le \mathsf{poly}(\secpar)$. It follows that $$\prob{\sseufcma_\ds(\sB,\secpar)=1} \ge \frac{S}{ \mathsf{poly}(\secpar)}$$ as needed.
\end{proof}

\subsection{\textsf{UD-qCPA} Security}\label{sssec:udqcpa}

In this section, we show how to transform a \sudqcpank extendable scheme into a \sudqccank secure scheme using digital signatures.

In order to obtain a \sudqcpank secure uncloneable decryptors, it suffices to provide \sudqcpank extendable uncloneable bit decryptors (recall \cref{def:extend}).

One might hope that the scheme \dcsobcca from \cref{con:dcs_cca1_res} is already \sudqcpank extendable. Unfortunately, this is not the case. As we will soon see, a poor choice of a copy protection scheme, even a \swqcp secure one, could result in a scheme that is not even $\sud^{1,1}$ extendable. Worse yet, by applying the transformation of \cref{con:dcs_cca2_res} to this scheme, we obtain a scheme which is \sobudqccaa secure but not $\sud^{1,1}$ extendable. To make things even worse, the scheme actually fails to be secure even when limited to plaintexts of length $2$!

\begin{proposition}\label{thm:ext_fail}
    Assume there exists a \swqcpnk (resp. \swqcp) secure copy protection scheme for some \bbf. Then there exists a \sobudqcpank (resp. \sobudqcpa) secure uncloneable decryption scheme which is not $\sud^{1,1}$ extendable even when limited to plaintexts of length $2$.
\end{proposition}

\begin{proof}
    The splitting attack described in \cref{sssec:qcp_split} shows that the existence of a \swqcpnk secure copy protection scheme for some \bbf implies the existence of a scheme which is also \swqcpnk secure, but with the property that each protected copy $\rho$ could be split into two states $\rho_0,\rho_1$ such that $\rho_b$ could be used to evaluate the underlying \bbf on inputs starting with $b$. Call that scheme \sqcp.
    
    By using this scheme to instantiate \cref{con:dcs_cca1_res} and then applying the transformation of \cref{con:dcs_cca2_res} to the result we obtain a scheme which is \sobudqccaank secure due to \cref{thm:con_1bit_dcs_cca1} and \cref{thm:dcs_cca2_res}. Let \ud be the scheme obtained by extending this scheme to two bit messages as described in \cref{def:extend}.
    
    We claim that this scheme is not $\ud^{1,1}$ secure. To see this, consider the following adversary $\sA = (\sP,\sD_0,\sD_1)$ to a version of the $\ud^{1,1}_\ud$ modified so that the challenge ciphers must be of length $2$ (note that we unusually named the distinguishers $\sD_0$ and $\sD_1$ rather than $\sD_1$ and $\sD_2$ for ease of notation):
    \begin{itemize}
        \item After being given $\rho$ from $\sC$, $\sP$ splits it to $\rho_0$ and $\rho_1$, and gives $\rho_b$ to $\sD_b$
        \item $\sD_b$ makes a challenge query on the plaintexts $m_\beta = \beta\|\beta$ for $\beta=0,1$. Recall that the ciphertext given to $\sD_b$ is of the form $(c_1,c_2)$ where $c_i = (r_i,\beta_i\oplus \bbf.\eval_\sk(r_i),\sigma)$, though $\sigma$ is currently irrelevant. Divide into cases:
        \begin{itemize}
            \item If $r_i$ starts with $b$ for some $i$ then $\sD_b$ uses $\rho_b$ to calculate $\beta_i \oplus \bbf.\eval_\sk(r_i)$ and outputs the result.
            \item Else, $\sD_b$ outputs a uniformly random bit.
        \end{itemize}
    \end{itemize}
    
    Note that since $r_i$ is uniformly random, the first case happens with probability $3/4$, and in this case $\sD_b$ outputs the correct answer with certainty. The second case has probability $1/4$ and then $\sD_b$ responds correctly with probability $1/2$. All and all, the probability that $\sD_b$ answers correctly is $\frac{3}{4} + \frac{1}{4}\cdot \frac{1}{2} = \frac{7}{8}$, so the expected number of distinguishers which answered correctly is $2\cdot\frac{7}{8} = 1 + \frac{1}{2} + \frac{1}{4}$, so that $\sA$ has a constant advantage of $\frac{1}{4}$.
\end{proof}

The good news, however, is that for \dcsobcca to be \sudqcpank extendable it suffices to require that the underlying copy protection scheme is \sfdqcpnk secure (recall \cref{def:fdqcp}).

\begin{proposition}\label{thm:dcs_cca1}
    Let \dcsobcca be obtained from instantiating \cref{con:dcs_cca1_res} with a \sfdqcpnk (resp. \sfdqcp) secure copy protection scheme (see  \cref{def:fdqcp}), then \dcsobcca is \sudqcpank (resp. \sudqcpa) extendable (see \cref{def:extend}).
\end{proposition}

\begin{proof}
    Let \sqcp be the copy protection scheme underlying \dcsobcca, and let \dcscpa be the result of applying the transformation of \cref{def:extend} to \dcsobcca. Recall that by \cref{thm:qcp-ria} and the \sfdqcpnk security of \sqcp it follows that \sqcp is also \sfdqcpriank secure (see \cref{def:qcp_ria}).
    
    Let $\sA=(\sP,\sD_1,\ldots,\sD_{n+k})$ satisfy that 
    $$\E\left[\sudqcpank_{\dcscca}(\sA)\right] = n + k/2 + \epsilon\text{,}$$ we construct $\sA'=(\sP',\sF'_1,\ldots,\sF'_{n+k})$ such that 
    $$\E\left[\sfdqcpriank_{\sqcp}(\sA')\right] = n + k/2 + \epsilon\text{.}$$ The \sfdqcpriank security of $\sqcp_\bbf$ will then imply that $\epsilon = \negl$ as needed.
    
    We first note that we can assume without loss that $m_0^i = \overline{m_1^i}$ for $i=1,\ldots,n_k$ (recall that $(m_0^i,m_1^i)$ is the $i$th pair of plaintexts sent to $\sC$ in \ref{step:dcs_chal} of \cref{def:dcs-indx}). This holds because by the very definition of \dcscca, each bit is encrypted independently of the rest of the bits. Assume for example that the first bit of both messages is $0$, let $\tilde{m}_b$ be obtained from $m_b$ by removing the first bit. Instead of sending $m_0,m_1$ in the challenge query, $\sD_i$ could:
    \begin{itemize}
        \item query the encryption oracle on input $0$ to obtain some cipher $c_1$,
        \item make a challenge on $\tilde{m}_0,\tilde{m}_1$ to obtain a cipher $(c_2,\ldots,c_\ell)$,
        \item run as before on the cipher $c=(c_1,\ldots,c_\ell)$.
    \end{itemize}
    The cipher $c$ obtained as above distributes exactly the same as the output of a challenge query on $m_0,m_1$. This process could be done in parallel for any bit that is the same in both messages, proving that an adversary restricted such that $m_0 = \overline{m_1}$ can achieve the same advantage as a general adversary. 
    
    We can therefore assume that $\sP$ outputs for each distinguisher a single message $m^i$, and that $\sD_i$ is given $c_b$ where $c_0$ is the encryption of $m^i$ and $c_1$ is the encryption of $\overline{m^i}$. $\sD_i$ then has to guess $b$.
    
    Assume $\sA = (\sP,\sD_1,\ldots,\sD_{n+k})$ is an adversary for the $\sudqcpank_{\dcscpa}$ game modified as described in the previous paragraph. We construct an adversary $\sA'= (\sP',\sF'_1,\ldots,\sF'_{n+k})$ for the $\sfdqcpriank_\sqcp$ game:
    \begin{itemize}
        \item after being given $\rho^{\otimes n}$ from $\sC$, $\sP'$:
        \begin{itemize}
            \item simulates $\sP(\rho^{\otimes n})$,
            \item responds to encryption calls on plaintexts of length $\ell$ by querying the random input oracle $\ell$ times to obtain $(r_i,b_i)$ for $i=1,\ldots,\ell$ and applying the unitary $$\ket{m,y}\mapsto\ket{m,y\oplus ((r_1,m_1\oplus b_1),\ldots,(r_\ell,m_\ell\oplus b_\ell))}$$ to the input,
            \item resumes the simulation until obtaining the states $\sigma_I$, and plaintexts $m^i$
            \item provides $\sigma_i\otimes m^i$ to $\sF'_i$.
        \end{itemize}
        \item after being given $\sigma_i'$, the freeloader $\sF'_i$:
        \begin{itemize}
            \item $\sF'_i$ queries $\sC$ on the bits of $m^i$ in order to obtain a sequence of pairs $$c_i = ((r_1,\beta_1),\ldots,(r_\ell,\beta_\ell))\text{,}$$
            \item simulates $\sD_i(\sigma_i)$ with $c_i$ as the cipher,
            \item responds to encryption queries exactly the same way as $\sP'$,
            \item resumes simulation of $\sD_i$ until obtaining an output $b_i$, and
            \item outputs $b_i$.
        \end{itemize}
    \end{itemize}
    
    The view of $\sD_i$ is exactly the same in the modified $\sudqcpank$ game and the simulation above, so they (and thereby $\sF_i$) output $b$ with the same probability. Hence it suffices to show that $b$ is the correct output for $\sD_i$ iff it is the correct output for $\sF'_i$.
    
    Indeed, if $\sC$ outputs pairs of the form $(r_i,\bbf.\eval_\sk(r_i))$ then the correct output for $\sF_i'$ is $0$. In this case $$\begin{aligned} & ((r_{1},\beta_{1}\oplus m_{1}),\ldots,(r_{\ell},\beta_{\ell}\oplus m_{\ell}))\\
= & ((r_{1},\bbf.\eval_\sk(r_1)\oplus m_{1}),\ldots,(r_{\ell},\bbf.\eval_\sk(r_\ell)\oplus m_{\ell}))\\
= & \dcscca.\enc_\sk(m)
\end{aligned}
$$
    so the correct output for $\sD_i$ is also $0$. The argument for the second case is identical.
\end{proof}

\subsection{\textsf{UD-qCCA2} Security}\label{sssec:ind_cca2}

Obtaining $\sudqccaa$ security requires a bit more care. The bit by bit approach does not prohibit basic manipulations of the ciphertext such as truncating, rearranging, or combining several qubits. This allows the adversary to slightly manipulate the challenge cipher such that the decryption of the new cipher completely reveals which message was encrypted in the challenge phase. Such attacks demonstrate that no scheme could be $\sudqccaa^{1,1}$ extendable.

This is overcome by signing a document containing the cipher and some metadata: a unique serial number sampled uniformly at random, the plaintext length, and the encrypted bit's location within the plaintext. These data make it infeasible to truncate, rearrange or combine ciphers to generate new ciphers. This allows us to generically transform a \sudqcpank extendable scheme (such as the scheme described in \cref{thm:dcs_cca1}) to a \sudqccaank secure scheme.

\begin{construction}[\sudqccaa Uncloneable Decryptors from \sudqcpa extendable Uncloneable Decryptors and \sseufcma Secure Signatures]\label{con:dcs_cca2}
    Let \dcsobcca be an uncloneable decryptors bit encryption scheme, and let \ds be a deterministic digital signature scheme.
    
    Let $\dcsccaa$ be the following scheme:
    \begin{itemize}
        \item $\dcsccaa.\keygen(\secparam)$ outputs $(\sk_\ud,\sk_\ds,\pk_\ds)$ where:
        \begin{itemize}
            \item  $\sk_\ud\gets \dcsobcca.\keygen(\secparam)$, and
            \item $\sk_\ds,\pk_\ds\gets \ds.\keygen(\secparam)$.
        \end{itemize}
        
        \item $\dcsccaa.\decgen(\sk_\ud,\sk_\ds,\pk_\ds)\equiv \dcsobcca.\decgen(\sk_\ud)$.
        \item $\dcsccaa.\enc_{(\sk_\ud,\sk_\ds,\pk_\ds)}(m)\to (r,(c_1,s_1),\ldots,(c_{|m|},s_{|m|}))$ where:
        \begin{itemize}
            \item $r\gets\{0,1\}^\secpar$,
            \item $c_i \gets \dcsobcca.\enc_{\sk_\ud}(m_i)$, and
            \item $s_i \gets \ds.\sign_{\sk_\ds}(c_i,|m|,i,r)$.
        \end{itemize}
        \item $\dcsccaa.\dec_{(\sk_\ud,\sk_\ds,\pk_\ds)}((r,(c_1,s_1),\ldots,(c_\ell,s_\ell))$ outputs
            $$\begin{cases}
                b_1\|\ldots\|b_\ell & \comma{\forall i=1,\ldots,\ell}{\ds.\ver_{\pk_\ds}((c_i,\ell,i,r),s_i) = 1}\\
                \bot & \mbox{else} \\
            \end{cases}$$ where $b_i = \dcsobcca.\dec_{\sk_\ud}(c_i)$.
        \item $\dcsccaa.\qdec(\rho\otimes \pk_\ds,(r,(c_1,s_1),\ldots,(c_\ell,s_\ell))$ outputs
            $$\begin{cases}
                b_1\|\ldots\|b_\ell & \comma{\forall i=1,\ldots,\ell}{\ds.\ver_{\pk_\ds}((c_i,\ell,i,r),s_i) = 1}\\
                \bot & \mbox{else} \\
            \end{cases}$$ where $b_i = \dcsobcca.\qdec(\rho,c_i)$.
    \end{itemize}
\end{construction}

\begin{proposition}\label{thm:udccaank}
    If \dcsobcca is \sudqcpank (resp. \sudqcpa) extendable (see \cref{def:extend}) and \ds is \sseufcma secure (see \cref{def:euf-qcma}) then the scheme \dcsccaa of \cref{con:dcs_cca2} is \sudqccaank (resp. \sudqccaa) secure (see \cref{def:dcs-indx}).
\end{proposition}

\begin{proof}
    Let \dcscca be the bit by bit extension of \dcsobcca to an unrestricted scheme as described in \cref{def:extend}. By hypothesis, this scheme is \sudqcpank secure.
    
    Let $\sA=(\sP,\sD_1,\ldots,\sD_{n+k})$ be an adversary to the $\sudqccaank_{\dcsccaa}$ game. We construct an adversary $\sA'=(\sP',\sD'_1,\ldots,\sD'_{n+k})$ such that
    $$\E\left[\left|\sudqccaank_{\dcsccaa}(\sA,\secpar) - \sudqcpank_{\dcscpa}(\sA',\secpar)\right|\right] < \negl \text{.}$$ The \sudqcpank security of \dcscpa will then imply the \sudqccaank security of \dcsccaa.
    
    After being given $\rho^{\otimes n}$ from $\sC$, the pirate $\sP'$:
    \begin{itemize}
        \item generates $\sk_\ds,\pk_\ds\gets \ds.\keygen(\secparam)$.
        \item simulates $\sP((\rho\otimes\pk_\ds)^{\otimes n})$, responding to an encryption queries of length $\ell$ by:
        \begin{itemize}
            \item using her encryption oracle $\ell$ times on input $0$ and $\ell$ times on input $1$ to obtain $\{c_{j,0},c_{j,1}\}_{j=1}^\ell$,
            \item sampling $r\gets\zo^\secpar$,
            \item computing $s_{j,b} \gets \ds.\sign_{\sk_\ds}(c_{j,b},\ell,j,r)$,
            \item storing $(s_{j,b},c_{j,b},b,\ell,r)$ to a list $L_\sP$, and
            \item applying the unitary $$\ket{m,y}\mapsto \ket{m,y\oplus(r,(c_{1,m_1},s_{1,m_1}),\ldots,(c_{\ell,m_\ell},s_{\ell,m_\ell}))}$$ to the input,
        \end{itemize}
        and decryption queries by applying the unitary $\ket{c,y} \mapsto \ket{c,t\oplus f(c)}$ where 
        $$f((r,(c_1,s_1),\ldots,(c_\ell,s_\ell))) = $$ $$
        \begin{cases} 
            b_1\|\ldots\|b_\ell & (c_1,s_1,b_1,\ell,r),\ldots,(c_\ell,s_\ell,b_\ell,\ell,r)\in L_\sP\\
            \bot  & \mbox{otherwise}
        \end{cases}\text{.}$$ (we implicitly assume that $f$ also returns $\bot$ on strings which are not of the required form.) Note that $f$ is well defined since it is impossible that $c_j$ is a cipher of both $0$ and $1$.
        \item when the simulation has ended, $\sP$ outputs for each $i$ the state $\sigma_i$ and a pair $(m^i_0,m^i_1)$ of plaintexts of length $\ell_i$, $\sP'$ gives each $\sD'_i$ the state $\sigma_i\otimes \pk_\ds \otimes \sk_\ds \otimes L_\sP$, and transmits the plaintext pairs to $\sC$.
    \end{itemize}
    
    The distinguisher $\sD'_i$ is given from $\sC$ a cipher of the form $c^i_1,\ldots,c^i_{\ell_i}$. She samples $r\gets \zo^\secpar$ and for $j=1,\ldots,\ell_i$ computes $s_i \gets \ds.\sign_{\pk_\ds}(c^i_j,\ell,j,r)$. She then simulates $\sD_i(\sigma_i)$ with $(r,(c_1,s_1),\ldots,(c_{\ell_i},s_{\ell_i})$ as the challenge cipher.
    
    Each distinguisher $\sD'_i$ simulates $\sD_i(\sigma_i\otimes \pk_\ds)$. $\sD'_i$ answers encryption calls the same way $\sP'$ did, storing signature-cipher-plaintext-length-nonce quintuplets into a list $L_i$ which we assume is initially a copy of $L_\sP$. She answers decryption calls the same way as well, only using the list $L_i$ rather than $L_\sP$ in the definition of $f$.

    Each distinguisher $\sD'_i$ simulates $\sD_i(\sigma_i\otimes \pk_\ds)$. $\sD'_i$ answers oracle calls the same way $\sP'$ did, storing signature-cipher-plaintext-length-noce quintuplets into a list $L_i$ which we assume is initially a copy of $L_\sP$.
    
    When $\sD$ makes a challenge query on plaintexts $m_0,m_1$ of length $\ell$ (that is, the simulation arrives at \cref{step:dcs_chal} in \cref{def:dcs-indx}), $\sD'$ forwards $m_0,m_1$ to $\sC$ as a challenge query to obtain a cipher of the form $c_1,\ldots,c_\ell$. She then samples $r\gets\zo^\secpar$ and responds to the challenge query made by $\sD$ with $(r,(c_1,s_1),\ldots,(c_\ell,s_\ell))$ where $s_j = \ds.\sign_{\sk_\ds}(c_j,\ell,j,r)$.
    
    We argue that if the inequality $$\E\left[\left|\sudqccaank_{\dcsccaa}(\sA,\secpar) - \sudqcpank_{\dcscpa}(\sA',\secpar)\right|\right] < \negl$$ does not hold, then it is possible to create an adversary which wins the $\sseufcma_\ds$ game (recall \cref{def:euf-qcma}) with non-negligible probability, by way of contradiction.
    
    As explained in the proof of \cref{thm:dcs_cca2_res}, if the actual and simulated views of $\sA$ by the end of the game are not statistically close, then at some point $\sA$ must have made a decryption query on an input whose result on the actual oracle is significantly different than on the simulated oracle. This can only happen if the input to the decryption query is significantly supported on legitimate ciphers which were not generated as a response to encryption calls and are not the challenge cipher. As a consequence, by measuring a random query in the computational basis, we obtain such a cipher with a non-negligible probability. The analysis is identical to the proof of \cref{thm:dcs_cca2_res}, so we do not repeat it here.
    
    It remains to explain why such a cipher necessarily contains a fresh signed document. Assume $c=(r,(c_1,s_1),\ldots,(c_\ell,s_\ell))$ is such a ciphertext. Then for any $i=1,\ldots,\ell$ it holds that $\ds.\ver_{\pk_\ds}(c_i,\ell,i,r) = 1$
    
    We assume that $\sA'$ never responds to encryption queries with two ciphers with the same nonce $r$. In practice, this only holds up to negligible probability (if this does happen, and the two ciphers happen to be encryption of plaintexts of the same size, then $\sA$ can indeed create a legitimate cipher which was not the output of any encryption query without creating any fresh signatures).
    
    Consider the following cases:
    \begin{itemize}
        \item One of the ciphers which $\sA'$ output as response to an encryption query is of the form $c'=(r,(c_1',s_1'),\ldots,(c'_{\ell'},s'_{\ell'}))$ (that is, $c$ and $c'$ have the same nonce $r$), then
        \begin{itemize}
            \item If $\ell\ne \ell'$ then it holds for any $i$ that the message $(c_i,\ell,i,r)$ was never signed by $\sA'$ so $((c_i,\ell,i,r),s_i)$ is a fresh signed document.
            \item Else, there exists some $i$ such that $c_i\ne c'_i$ or $s_i \ne s'_i$ then $((c'_i,\ell,i,r),s_i)$ is a fresh signed document (note that here we use the fact that \ds is \sseufcma secure and not just \textsf{EUF-CMA} secure. \textsf{EUF-CMA} security does not prevent the adversary to create a different signature for the same message whereby modifying query outputs to create legitimate ciphertexts).
        \end{itemize}
        \item Else, each pair of the form $((c_i,\ell,i,r),s_i)$ is a fresh signed document.
    \end{itemize}
\end{proof}
%auto-ignore
\section{Open Questions}\label{sec:disc}

Our treatment raises several questions. We list some of them here:

\paragraph{Can \textsf{FLIP-QCP} security be obtained from \textsf{WEAK-QCP} security?} We have shown that \swqcp security does not imply \sfdqcp security. Is it possible to generically transform a \swqcp secure scheme into a \sfdqcp secure scheme? Such a transformation will show that our unrestricted length schemes could be instantiated from \swqcp security.

\paragraph{Can our construction be made more generic?} Is there a generic transformation of \sudqcpa uncloneable decryptors into \sudqcca secure uncloneable decryptors, perhaps utilizing additional primitives? In particular, does \sudqx security imply \sudqx extendability? A positive answer to the latter would imply that \sudqccaa security can be generically obtained from \sudqcpa security and \sseufcma digital signatures.

\paragraph{Semantic security notions?} Our security notions extend the notions of indistinguishability of ciphertexts security for encryption schemes. It is known that in the setting of encryption schemes, the indistinguishability of ciphertexts is equivalent to a more natural concept of security known as \emph{semantic definition} which better captures the adversary's limited ability to learn \emph{anything} about a plaintext given its encryption. Our definition would be better established if we could provide a semantic security definition and prove its equivalence to the current definition. 

There are also several natural directions into which the current work could be extended:
\begin{itemize}
    \item Public key encryption: could the notion of \sudqccaa be extended to the public key setting? The authors of \cite{ALLZZ20} consider single decryptors in the public key setting where the adversary has no oracle access, which amounts to an asymmetric notion of \sudqcpa.
    \item Quantum challenge queries: our security definitions extend the \sindqx security notions of symmetric encryption afforded by \cite{BZ13}. In these definitions, oracle queries might be in superposition, but the challenge query is classical. The works of \cite{GHS16} and \cite{CEV20} afford notions of \sqindqcpa and \sqindqccaa where the challenge phase is also in superposition. It is interesting whether their definitions could be extended to stronger security notions for uncloneable decryptors and whether such decryptors could be constructed.
    \item Quantum encryption: our syntax, security, and constructions are suitable for encrypting classical plaintexts. The authors of \cite{ABF+16} extend the \sindcx security notions to the setting where the plaintexts are arbitrary quantum states. Can these definitions be extended to support uncloneable decryptors?
\end{itemize}

%%%%%%%%%%%%%%%%%%%%%%%%%%%%
%%%%%%%%%%%%%%%%%%%%%%%%%%%%
%NON-ANON PART

\ifnum\anon=0
%ANON PART
\fi

\ifnum\sigconf=1
    \bibliographystyle{ACM-Reference-Format}
\else
    \bibliographystyle{alphaabbrurldoieprint}
\fi

\ifnum\masterthesis=0
    \bibliography{main}
\fi

\appendix
\ifnum\shownomenclature=1
\printnomenclature[1in]
%There is a label which is created automatically in the macros of the following form: \label{sec:nomenclature}
\fi
\setcounter{tocdepth}{0}
%auto-ignore
\section{Splitting Attack on Weak \textsf{QCP}}\label{app:split}

In \cref{sssec:qcp_split} we informally described a splitting attack against $\swqcp^{1,1}$ secure copy-protection scheme. In this appendix we formalize this attack. That is, we show how to transform a $\swqcp^{1,1}$ secure copy-protection scheme for \bbf to a new scheme which is $\swqcp^{1,1}$ secure (though for a slightly different \bbf) but splittable. Our attack is easily extendable to \swqcpnk security.

Let $\sqcp_\bbf$ be a $\swqcp^{1,1}$ secure copy-protection scheme for the balance binary function $\bbf$.

We define a new binary function $\bbf'$ as follows:
\begin{itemize}
    \item $\bbf'.\samp(\secparam)$: invoke $\bbf.\samp(\secparam)$ twice to obtain two functions $f_0,f_1$, output $(f_0,f_1)$,
    \item $\bbf'.\eval((f_0,f_1), b\|x)$: output $\bbf.\eval(f_b,x)$.
\end{itemize}

It is trivial to see that $\bbf'$ is also balanced.

We define the following copy-protection scheme $\sqcp'.\bbf'$:
\begin{itemize}
    \item $\sqcp'.\ptect(f_0,f_1)$: output $\rho_{(f_0,f_1)} =\rho_0\otimes \rho_1$ where $\rho_b \gets \sqcp.\ptect(f_b)$,
    \item $\sqcp'.\qeval(\rho_{(f_0,f_1)},b\|x)$: output $\sqcp.\qeval(\rho_b,x)$
\end{itemize}

Note that this scheme affords a very simple splitting attach. The state $\rho_{(f_0,f_1)}$ is already given as a tensor product of the states $\rho_{f_b}$, each of which could be used to evaluate $f$ on points whose first bit is $b$. It remains to show that the new scheme $\sqcp'$ is also $\swqcp^{1,1}$ secure, from which will follow that weak copy-protection security does not prohibit splitting attacks.

\begin{claim}
    $\sqcp'$ is a $\swqcp^{1,1}$ secure copy-protection scheme for $\bbf'$.
\end{claim}

\begin{proof}\label{claim:splitting}
    Let $\sA'=(\sP',\sF'_1,\sF'_2)$ be an adversary for the $\swqcp^{1,1}_{\sqcp'}$ game with advantage $\epsilon$. We need to show that $\epsilon = \negl$. To do so, we construct an adversary $\sA=(\sP,\sF_1,\sF_2)$ for the $\swqcp_{\sqcp}$ game with advantage $\epsilon$. It will then follow by hypothesis that $\epsilon = \negl$.
    
    The pirate $\sP$:
    \begin{itemize}
        \item samples a random bit $b$, let $\rho_b$ be the state afforded by $\sC$,
        \item samples $f\gets \bbf.\samp(\secparam)$ and $\rho_{1-b}\gets \sqcp.\ptect(f)$,
        \item starts simulating the $\swqcp_{\sqcp'}(\sA')$ game by affording the pirate $\sP'$ with the state $\rho_0\otimes \rho_1$ until they obtain from $\sP'$ the states $\sigma_1$ and $\sigma_2$,
        \item gives $\sigma_i \otimes b$ to $\sF_i$.
    \end{itemize}
    
    The freeloader $\sF_i(\sigma_i\otimes b,b\|x)$ simulates $\sF'_i(\sigma_i,x)$  until they obtain an output $y_i$, which she outputs herself.
    
    Note that $y_i = f_b(x)$ if and only if $y_i = \bbf'.\eval((f_0,f_1),b\|x)$, so the probability that $\sF_i$ evaluated $x$ correctly is exactly the same as the probability that $\sF'_i$ evaluated $b\|x$ correctly. However, by definition of $\bbf'$ and $\sA$ it follows that the view of $\sA'$ in the simulation above is identical to her view in the $\swqcp_{\sqcp'}(\sA')$ game, hence the advantage of $\sA$ for the $\swqcp_{\sqcp}$ game is also $\epsilon$.
\end{proof}

Finally, we use the splitting attack above to prove that the converse of \cref{lem:fdtow} is false:

\begin{lemma}\label{lem:fdfail}
    If there exists a $\swqcp^{1,1}$ secure copy-protection scheme, then there exists a copy-protection scheme which is $\swqcp^{1,1}$ secure but not $\sidqcp^{1,1}$ secure.
\end{lemma}

\begin{proof}
    Let $\sqcp$ be $\swqcp^{1,1}$ secure, and let $\sqcp'$ be the scheme described in \cref{sssec:qcp_split}. According to \cref{claim:splitting}, $\sqcp'$ is also $\swqcp^{1,1}$ secure. However, it is not $\sidqcp^{1,1}$ secure, as we will now show by directly constructing an adversary $\sA = (\sP,\sF_0,\sF_1)$ for that $\sidqcp^{1,1}(\sA,\secpar)$ game (note that we unusually index the freeloaders with $0,1$, as it is more convenient in our particular context).
    
    The pirate $\sP$, upon being given $\rho_{f_0,f_1} = \rho_{f_0} \otimes \rho_{f_1}$ gives $\rho_{f_\beta}$ to $\sF_\beta$ resp.
    
    The freeloader $\sF_\beta$ queries $\sO_{f,b_\beta}$ twice to get two pairs $(x_1,f(x_1)\oplus b_\beta),(x_2,f(x_2))\oplus b_\beta$. If the string $x_1$ starts with $\beta$, let $\tilde{x}_1$ be $x_1$ without the first bit, then $\sF_\beta$ calculates $f(x_1) \gets \sqcp.\qeval(\rho_\beta,\tilde{x}_1)$ and outputs $f(x_1)\oplus (f(x_1)\oplus b_\beta) = b_\beta$. Else, if the string $x_2$ starts with $\beta$, they do the same with $x_2$. In both cases they output $b_\beta$ with certainty. In case both strings $x_1$ and $x_2$ start with $1-\beta$, $\sF$ outputs a uniformly random bit.
    
    Since $x_1,x_2$ are uniformly random, the probability that both start with $1-\beta$ is $1/4$. It follows that $\sF_\beta$ outputs the correct answer with probability $3/4 + 1/4\cdot 1/2 = 7/8$. Hence, $\E\left[\sidqcp^{1,1}(\sA,\secpar)\right] = 14/8 = 1 + 1/2 + 1/4$, so that $\sA$ has a non-negligible advantage of $1/4$.
    
    Note that by doing the same but with more queries, the advantage becomes exponentially close to the maximal advantage of $1/2$.
\end{proof}

\section{\textsf{FLIP-QCP} as a Copy-protection Version of \textsf{LoR-CPA}}\label{app:lor_qcp}

The purpose of this section is to better motivate the notion of flip detection security for quantum copy-protection (see \cref{def:fdqcp}) introduced and considered in \cref{ssec:qcp_ids}. We do so by arguing that it is a natural adaptation of the notions used to model multi-encryption security for encryption schemes. In this appendix, we overview the left and right notions of security for encryption schemes and present an adaptation of this definition to copy-protection scheme we call \slorqcp security. We show that \slorqcp security implies \sfdqcp security for copy-protecting $\bbf$s. We also show that \slorqcp implies \swqcp in general.

The motivation for \cref{def:fdqcp} arose through attempts to extend uncloneable bit decryptors to uncloneable decryptors which support messages of arbitrary length. A standard method to extend bit encryption to arbitrary messages is to encrypt each bit separately. In order to establish the security of the resulting scheme, it is required to show that no new attacks arise from the fact that a cipher of a long plaintext could be parsed as many independently encrypted bits. Differently stated, it should be established that the original bit encryption scheme remains secure against an adversary with access to \emph{multiple encryptions}. Such adversaries are usually modeled with notions of \emph{left or right} (\slor) security (for example, in \cref{def:lorcpa} we introduce the notion of \slorcpa, though in a different context). Informally, the difference between \sind (see \cref{ssec:pre_ses}) and \slor security lies in the form of the challenge phase. In the \sind notions, the adversary is aware of two plaintexts, which we can call the \emph{left plaintext} and \emph{right plaintext}. She is given a ciphertext and has to distinguish whether it is an encryption of the left or the right plaintext. The \slor notions of security allow the presence of \emph{many} pairs of left and right ciphers. The challenger uniformly chooses a side, and for each pair affords the adversary with a ciphertext of the plaintext of that side (the crucial point being that the challenger chooses a side only once, and consistently encrypts only the plaintexts of that side throughout the game). The adversary then has to distinguish a challenger who chose the left side from a challenger who chose the right side.

It is easy to see that \slor security is at least as strong as \sind with the same oracle access (as \sind security could be described as \slor security where the adversary is limited to a single pair of left and right plaintexts). Basic results in encryption theory show that the inverse implication is also true (albeit at the cost of increasing the adversary's advantage by a polynomial multiplicative factor) in contexts where the adversary has access to an encryption oracle. These statements are formally defined, proved, and discussed in \cite{BDJR97}. Unfortunately, this property of encryption schemes does not hold for uncloneable decryptors, as is demonstrated in \cref{thm:ext_fail}. However, \cref{thm:dcs_cca1} shows that it is possible to obtain an uncloneable bit decryptor scheme that remains secure under multiple encryptions if we require the underlying copy-protection scheme to satisfy the flip detection security notion defined in \cref{def:fdqcp}.

In this appendix, we present a natural adaptation of left or right security to copy-protection schemes we call \slorqcp; we then show that \slorqcp implies \swqcp security in general, and \sfdqcp security for balanced binary function. A subtlety with this definition is with the distribution $\sD$ with respect to which it is defined. We restrict our attention to the scenario where all inputs are sampled independently, though it is possible to define \slorqcp in more general settings.

\begin{definition}[\slorqcp security]\label{def:lorqcp}
    Let $\sF_\secpar$ be a permissible class of functions (see \cref{def:admiss}). Let $\sD_\sF$ be a polynomial randomized circuit such that $\sD_\sF(\secparam)$ outputs an element of $\sF_\secpar$. Let $\sD_x$ be a polynomial randomized circuit such that for any $f\in \sF$, $\sD_x(f)$ outputs an element in the domain of $f$.
    
    Denote $\sD = (\sD_\sF,\sD_x)$. Let $\sqcp$ be a copy-protection scheme for $\sF$ (see \cref{def:qcp_synt}).
    
    For any $f\in\sF$ and $b\in\zo$ let $\sO_{f,b}$ be the oracle which is given no input and outputs $(x_0,x_1,f(x_b))$ where $x_0\gets \sD_x(f)$, and $x_1$ is sampled from $\sD_x(f)$ conditioned on the event that $f(x_0)\ne f(x_1)$.
    
    For any $n,k = \poly$ and for any adversary $\sA = (\sP,\sF_1,\ldots,\sF_{n+k})$ define the $\slorqcp^{n,k,\sD}_{\sqcp}(\sA,\lambda)$ game between $\sA$ and a trusted adversary $\sC$:
    
    \begin{itemize}
        \item $\sC$ samples $f\gets \sD_\sF$, invokes $\qcp.\ptect(f)$ $n$ times to obtain $\rho=\rho_f^{\otimes n}$ and gives $\rho$ to $\sP$,
        \item $\sP$ generates $n+k$ states $\sigma_1,\ldots,\sigma_{n+k}$ and gives $\sigma_i$ to $\sF_i$,
        \item $\sC$ samples $b_i\gets \zo$ for $i=1,\ldots,n+k$,
        \item $\sF_i^{\sO_{f,b_i}}(\sigma_i)$ outputs a bit $\beta_i$,
        \item the output of the game is the number of indices $i$ such that $b_i=\beta_i$.
    \end{itemize}
    
    $\qcp$ is \emph{$\slorqcp^{n,k,\sD}$ secure} if for any $\sA$ it holds that
    $$\E\left[\slorqcp^{n,k,\sD}_{\sqcp}(\sA,\lambda)\right] \le n + \frac{k}{2} + \negl\text{,}$$ $\qcp$ is \emph{$\slorqcp^{\sD}$ secure} if it is $\slorqcp^{n,k,\sD}$ secure for any $n,k=\poly$.
\end{definition}

We soon prove that a \slorqcp secure copy-protection scheme for a \bbf is also \sfdqcp secure. In particular, the \slorqcp security of a scheme $\qcp_\bbf$ implies that \cref{con:dcs_cca1_res} is \sudqcpa extendable when instantiated with $\qcp_\bbf$. This implication can be interpreted as an indication that the difficulty in extending uncloneable bit decryptors is actually underwritten by a broader difficulty: that of extending the uncloneability of evaluating a function to the uncloneability of identifying to which of two sets of arbitrarily many random points this functionality was applied. This extension --- which is commonplace in encryption (see the discussions titled "Multiple-Message Security" in sections 5.4.3 and 5.4.4 in \cite{Gol04} for an overview) --- seems to fail in the context of quantum copy-protection.

\begin{remark}
    \slor security is defined like \sind security (recall \cref{def:ind-qccax}) with the modification that the adversary creates a list of polynomially many $q$ pairs of plaintexts $(m_0^1,m_1^1),\ldots,(m_0^q,m_1^q)$ and is then given back the encryptions of $m_b^1,\ldots,m_b^q$ and has to distinguish the cases $b=0$ from $b=1$ (note that all plaintexts are chosen at once, unlike in \cref{def:lorcpa} where the adversary is given the encryption of $m_b^i$ \emph{right after} providing the pair $(m_0^i,m_1^i)$, and may adapt future queries based on the ciphertext of previous queries).
    
    It is easy to construct an example where \sind security does not imply \slor security (for example, one could choose a random string $r$ the same length of the secret key, and then append $r$ to encryptions of $0$ and $\sk\oplus r$ to encryptions of $1$).
    
    One could argue that since copy-protection adversaries are not given oracle access, the lack of equivalence of \swqcp and \sfdqcp is similar to the lack of equivalence of \sind and \slor and has nothing to do with uncloneability.
    
    There are three responses to this argument:
    \begin{enumerate}
        \item We also see this phenomenon in the context of uncloneable decryptors, where the security game is more similar to those used to define \sindcpa security, and in particular affords encryption oracle access to the adversary. The failure of a \sobudqcpa secure scheme to be \sudqcpa extendable seen in \cref{thm:ext_fail} stands in contrast to the equivalence between \sindcpa and \slorcpa security.
        \item In both \sfdqcp and \swqcp, the adversary is given access to the copy-protected program, which allows them to evaluate the function at arbitrary points and is akin to having oracle access to the function at the first phase of the \sfdqcp and \swqcp games. As we establish in \cref{thm:qcp-ria}, this property is strong enough that allowing the adversary to evaluate the function in uniformly random points throughout the entirety of the game does not increase the security. The randomness of the evaluated points is a manifestation of the non-determinism of the encryption procedure (which is a necessary condition for \sindcpa security).
        \item Some may find the previous point objectionable, since the structure of the security games still seems quite different. This objection could be addressed by investigating the security notions for encryption obtained by allowing the adversary access to encryptions of \emph{random} plaintexts. Such security notions have not been researched in the literature to the best of our knowledge, as they do not seem particularly useful to model any forms of attack not already addressed by chosen plaintext security. However, it is straightforward to define \sind and \slor security against \emph{random} plaintext attacks and show that these notions are indeed equivalent.
    \end{enumerate}
\end{remark}

\begin{proposition}
    Let \bbf be a balanced binary function with input length $\ell$ (see \cref{def:bbf}), $\qcp$ be a copy-protection scheme for \bbf (see \cref{def:qcp_synt}), and let $\sD = (\bbf.\samp,\mathcal{U}(\zo^\ell))$. For any $n,k=\poly$ if $\qcp$ is $\slorqcp^{n,k,\sD}$ secure, then \qcp is also \sfdqcpnk secure.
\end{proposition}

\begin{proof}
    Given an adversary $\sA=(\sP,\sF_1,\ldots,\sF_{n+k})$ for \sfdqcpnk consider an adversary $\sA'=(\sP',\sF'_1,\ldots,\sF'_{n+k})$ for $\slorqcp^{n,k,\sD}$ which simulates $\sA$ and responds to queries as follows: whenever $\sF_i$ queries for a pair of the form $(x,f(x)\oplus b_i)$, $\sF'$ makes a query to $\tilde{\sO}_{f,b_i}$ to obtain a pair of the form $(m_0,m_1,f(m_{b_i}))$ and responds to the query made by $\sF_i$ with $(m_0,f(m_{b_i}))$. Finally, $\sF'_i$ outputs the output of $\sF_i$. By the fact that $f(m_0)\ne f(m_1)$ it follows that $(m_0,f(m_{b_i})) = (m_0,f(m_0)\oplus b_i)$ so that $\sF_i'$ answers correctly if and only if $\sF_i$ answers correctly.
\end{proof}

We conclude by showing that \slorqcp security with respect to any distribution $\sD$ implies \swqcp security with respect to the same distribution.

\begin{lemma}\label{thm:lorwqcp}
    Let \sqcp be a quantum copy-protection function that is $\slorqcp^{n,k,\sD}$ secure with respect to some distribution $\sD$. Then it is also $\swqcp^{n,k,\sD}$ secure.
\end{lemma}

\begin{proof}
    Let $\sA = (\sP,\sF_1,\ldots,\sF_{n+k})$ be an adversary whose expected output in the $\swqcp^{n,k,\sD}$ game is $\mu$. It is easy to construct an adversary $\sA = (\sP,\sF_1',\ldots,\sF'_{n+k})$ whose expected output in the $\slorqcp^{n,k,\sD}$ game is at least $\mu$. The freeloader $\sF'_{i}$:
    \begin{itemize}
        \item makes a single query to $\sO_{f,b}$ to obtain $((x_0,x_1),y)$,
        \item invokes $\sF_i$ with input $x_0$,
        \item outputs $0$ if and only if $\sF_i$ outputs $y$.
    \end{itemize}
    
    It is straightforward to check that if $\sF_i$ evaluates $f(x_0)$ correctly, then $\sF_i'$ outputs $b$ (here we rely on the fact that $f(x_0) \ne f(x_1)$). So the probability that $\sF_i'$ is correct is at least the probability that $\sF_i'$ is correct. Note that if $b=1$ then $\sF_i'$ outputs $b$ even if $\sF_i$ evaluated $f(x_0)$ wrong, as long as she did not happen to output the point $y$, which means that the expected output of $\sA'$ in the $\slorqcp^{n,k,\sD}$ game might actually be larger (albeit this scenario is impossible if $\sF$ is a binary function).
\end{proof}

\section{Oracle Instantiation of \textsf{FLIP-QCP}}\label{app:flip_unlearn}

Aaronson \cite{Aar09} introduces a quantum oracle relative to which there exists a \swqcp secure copy-protection scheme for any unlearnable class of functions. In a subsequent work, Aaronson et al. \cite{ALLZZ20} manage to replace the quantum oracle with a classical oracle at the cost of only obtaining $\swqcp^{1,1}$ security.

\begin{remark}
    The reason that \cite{Aar09}'s construction supports unlimited copies whereas \cite{ALLZZ20}'s only support a single copy follows from the information theoretic no-go theorem to which they eventually appeal. The \cite{Aar09} construction encodes programs into Haar random quantum states, which are hard to clone approximately even when given (polynomially) many copies. The \cite{ALLZZ20} construction replaces the Haar random state with so called \emph{hidden subspaces} first introduced in \cite{AC12} to construct quantum money. These states have the property that they could be verified given access to an appropriate \emph{classical} oracle. The \cite{ALLZZ20} reduction appeals to a theorem by \cite{BS16} which shows that a particular functionality of hidden subspaces (namely, producing either a primal or a dual vector of the hidden subspace) can not be cloned. Unlike random states, this theorem no longer holds when an adversary is given linearly many copies of the hidden subspace. The \cite{ALLZZ20} scheme uses fresh randomness for each copy of the program, and thereby their scheme might still be secure in the presence of many copies. However, their proof is highly specialized for the single copy setting, and they leave the more general setting as an open problem.
\end{remark}

In this section, we sketch a proof that the \cite{Aar09} and \cite{ALLZZ20} schemes are actually \sfdqcp and $\sfdqcp^{1,1}$ secure respectively. We focus on \cite{ALLZZ20} since the authors thereof provide the notion of predicates that can naturally encapsulate more general functionalities we wish to copy-protect, including flip detection security. Our arguments also apply to the construction of \cite{Aar09} since they are mostly agnostic to the proof of \cite[Theorem~4]{ALLZZ20}, and only appeal to the fact that it generalizes to arbitrary predicates (see the discussion concluding Section 5 therein, which is equally applicable to \cite[Theorem~9]{Aar09}). As we shortly discuss, the only detail in our argument which requires addressing the content of the proof is justifying that it still holds even though the predicate we define require superpolynomially many random bits (note that the proof does \emph{not} follow through for such predicates in general). 

To extend this result to \sfdqcp security, we note two simple observations:
\begin{itemize}
    \item the oracle instantiated scheme of \cite{ALLZZ20} is \sfdqcp secure, given that the underlying function satisfies an ostensibly stronger form of unlearnability we shortly introduce called \emph{flip unlearnability}, and
    \item flip unlearnability is actually equivalent to unlearnability.
\end{itemize}

In order to consider more generalized notions of copy-protection, which seek to prevent the freeloaders from performing functionalities that are not as strong as evaluating the underlying function at a random point, \cite[Definition~15]{ALLZZ20} introduce the formalism of \emph{predicates}. A predicate allows us to describe a variety of functionalities by encapsulating both the information given to the freeloader and the outcomes which are considered "correct".

\begin{definition}[Predicate, \cite{ALLZZ20}]
    A \emph{binary predicate} $E=E(P,y,r)$ is a binary outcome function comprised of the following information:
    \begin{itemize}
        \item a deterministic circuit $X(y,r)$, and
        \item a relation $R$ which contains triplets of the form $(z,y,r)$.
    \end{itemize}
    Given a circuit $P$, the predicate $E(P,y,r)$ evaluates to $0$ if and only if $(P(X(y,r),y,r) \in R$.
\end{definition}

Here we assume that $P$ is a circuit with classical input and output, though it can be a quantum circuit with auxiliary qubits instantiated to an arbitrary quantum state. In the definitions below, $P$ represents a freeloader, and the auxiliary state is the input given to $P$ by the pirate.

While not implied by the definition of a predicate, the values $y$ and $r$ should be considered as representing arbitrary auxiliary information and uniform randomness, respectively. Given a predicate, we can redefine both unlearnability and copy-protection in terms of that predicate to obtain the notions of \emph{generalized copy-protection} and \emph{generalized unlearnability}, which are formally defined in Definitions 19 and 21 in \cite{ALLZZ20}. The idea is that instead of testing whether the adversary manages to evaluate the function in a random point, we test whether given input $x=X(y,r)$ (where $y$ is arbitrary auxiliary information which may depend on the underlying function $f$) they produce $z$ such that $(z,y,r)\in R$. This definition is rather involved, so we do not reproduce it here. Note that by setting $y=f$, $r\gets \zo^\ell$, and $(z,f,r)\in R \iff f(r)=z$ we recover the standard notions of unlearnability and copy-protection.

\begin{remark}
    The definitions of generalized unlearnability and generalized copy-protection afforded in \cite{ALLZZ20} are more general still. They are defined with respect to a \emph{cryptographic application} which is comprised of a predicate as well as a \emph{sampling protocol} between the challenger $\sC$ and the adversary $\sA$ for choosing the function $f$. Moreover, their sampling protocol allows passing arbitrary auxiliary information to the adversary. We do not require this level of generality and always implicitly assume that in the sampling "protocol" $\sC$ simply samples $f\gets \bbf.\samp(\secpar)$ without any communication with $\sA$.
\end{remark}

A minor difficulty in defining flip detection as a predicate is that the amount of randomness in the definition of a predicate may not depend on the adversary. Since an adversary can make any polynomial number of queries, each requiring a constant number of random bits to respond, this forces the randomness to be superpolynomially large. To sidestep this, we choose some fixed superpolynomial function $q$ and let $r$ be long enough to contain $q$ pairs of random inputs. When considering a circuit $P$ which makes $q'=\poly$ queries, we redefine it as a circuit that gets $q$ responses in advance, uses the first $q'$ as query responses, and discards the rest. The resulting predicate accounts for all possible QPT adversaries.

Allowing superpolynomial randomness seems to break the security reductions in the proof of \cite[Theorem~4]{ALLZZ20}, as they require a computationally bounded adversary to evaluate the predicate $E$. However, this is easily fixed by noting that for a \emph{fixed} adversary, the result of evaluating $E$ is independent of all but a polynomial prefix of the randomness, so it can be simulated efficiently.

Let $q$ be a superpolynomial function, we define the flip detection predicate as follows: the randomness is of the form $r=(b,r_1,\ldots,r_q)$, the input is $X(f,r)=((r_1,f(r_1)\oplus b),\ldots,(r_q,f(r_q)\oplus b))$, and $(z,f,r)\in R \iff z = b$.

We say that a \bbf is \emph{flip unlearnable} if it is satisfies \cite[Definition~21]{ALLZZ20} with respect to the flip detection predicate and the trivial sampler (we do not fully specify this definition as it would require us to introduce several notions form \cite{ALLZZ20} which are too far removed from the discussion. Fortunately, the fact that the functions under consideration are $\bbf$s allows for much simpler equivalent definitions for both unlearnability and flip unlearnability, which we introduce shortly). Applying the generalized form of \cite[Theorem~4]{ALLZZ20} gives us:

\begin{lemma}\label{lem:aota}
    If \bbf is flip unlearnable, then applying the \cite{Aar09} construction (resp. \cite{ALLZZ20} construction) to \bbf results in a \sfdqcp secure (resp. $\sfdqcp^{1,1}$ secure) copy-protection scheme.
\end{lemma}

\cref{lem:aota} already implies that uncloneable decryptors could be instantiated from a quantum oracle assuming the existence of a flip unlearnable \bbf. The remaining question is how quantum unlearnability and flip unlearnability compare. We prove that these are equivalent notions.

In order to do so, we first note that the definitions of unlearnability and flip unlearnability take on a simpler form when the underlying function is a \bbf.

\begin{definition}[unlearnability for \bbf]\label{def:ulbbf}
    Let \bbf be a balanced binary function with input length $\ell$, and let $\sA=(\sA_1,\sA_2)$ be a procedure. Define the $\sul_\bbf(\sA,\secpar)$ game between $\sA$ and a trusted challenger $\sC$ as follows:
    \begin{enumerate}
        \item $\sC$ samples $f\gets \bbf.\samp(\secparam)$,
        \item $\sA_1^{\ket{f}}$ produces a quantum state $\sigma$,
        \item $\sC$ samples $x\gets \zo^\ell$,
        \item $\sA_2(\sigma,x)$ outputs a bit $b$,
        \item the output of the game is $1$ if and only if $b=f(x)$ (in that case we say that $\sA$ \emph{won} the game).
    \end{enumerate}
    
    We say that \bbf is \emph{unlearnable} if it holds for any QPT $\sA$ that
    $$\E\left[\sul_\bbf(\sA,\secpar)\right] = \frac{1}{2} + \negl \text{.}$$
\end{definition}

\begin{definition}[flip unlearnability]
    For any binary function $f$ with input length $\ell$ and any bit $b$ let $\sO_{f,b}$ be an oracle with no input which outputs a pair of the form $(r,f(r)\oplus b)$ where $r\gets \zo^\ell$.
    
    Let \bbf be a balanced binary function with input length $\ell$, and let $\sA=(\sA_1,\sA_2)$ be a procedure. Define the $\sful_\bbf(\sA,\secpar)$ game between $\sA$ and a trusted challenger $\sC$ as following:
    \begin{enumerate}
        \item $\sC$ samples $f\gets \bbf.\samp(\secparam)$ and $b\gets\zo$,
        \item $\sA_1^{\ket{f}}$ produces a quantum state $\sigma$,
        \item $\sA_2^{\sO_{f,b}}(\sigma)$ outputs a bit $b'$,
        \item the output of the game is $1$ if and only if $b=b'$ (in that case we say that $\sA$ \emph{won} the game).
    \end{enumerate}
    
    We say that \bbf is \emph{flip unlearnable} if it holds for any QPT $\sA$ that
    $$\E\left[\sful_\bbf(\sA,\secpar)\right] = \frac{1}{2} + \negl \text{.}$$
\end{definition}

\begin{lemma}\label{thm:qul_qful}
    Let \bbf be a balanced binary function, then \bbf is unlearnable if and only if it is flip unlearnable.
\end{lemma}

\begin{proof}
    The easy direction is showing that flip unlearnability implies unlearnability. Say $\sA=(\sA_1,\sA_2)$ wins the $\sul$ game with probability $p$, we will construct a procedure $\sA_2'$ such that $\sA'=(\sA_1,\sA_2')$ which wins the $\sful_\bbf$ game with the same probability: $\sA_2'(\sigma)$ uses her oracle access to obtain a single pair of the form $(r,f(r)\oplus b)$, she then invokes $\sA_2(\sigma,r)$ to obtain an output $b'$. Finally, she outputs $f(r)\oplus b \oplus b'$. Note that $\sA'$ wins the game if and only if $f(r)=b'$. Since the view of $\sA_2$ is identical in the $\sul_\bbf$ game and in the simulation, it follows by hypothesis that $f(r)=b'$ with probability $p$. From the assumption that $\bbf$ is flip unlearnable it follows that $p = \frac{1}{2} + \negl$ as needed.
    
    For the other direction, assume that $\sA=(\sA_1,\sA_2)$ is an adversary to the $\sful$ game which makes $q$ queries to $\sO_{f,b}$. Let $H_i$ be the hybrid where the first $i$ queries are to the oracle $\sO_{f,0}$ and the remaining queries are to the oracle $\sO_{f,1}$, and let $p_i$ be the probability that $\sA$ outputs $0$ in the hybrid $H_i$.
    
    Note that $H_0$ is exactly the $\sful_\bbf$ game conditioned on $b=0$, and $H_q$ is exactly the $\sful_\bbf$ game conditioned on $b=1$, so to prove that \bbf is flip unlearnable it suffices to show that $|p_0-p_q| = \negl$.
    
    Consider the following adversary $\sA'=(\sA_1',\sA_2')$ to the $\sful_\bbf$ game:
    \begin{enumerate}
        \item $\sA_1'$ simulates $\sA_1$ to obtain a quantum state $\sigma$, she also uses her oracle to create $q$ pairs of the form $(r,f(r))$ with $r\gets \zo^\ell$. She creates the state $\sigma' = \sigma \otimes ((r_0,f(r_0)),\ldots,(r_q,f(r_q)))$.
        \item $\sA_2'(\sigma',x)$ samples $\iota\gets {1,\ldots,q}$, she then simulates $\sA_2$ and responds to the $i$th query the following way:
        \begin{itemize}
            \item for $i<\iota$ output $(r_i,f(r_i))$,
            \item for $i=\iota$ output $(x,0)$,
            \item for $i>\iota$ output $(r_i,f(r_i)\oplus 1)$.
        \end{itemize}
        $\sA_2'$ resumes the simulation until $\sA_2$ outputs a bit $b$, which $\sA_2'$ outputs herself.
    \end{enumerate}
    
    Note that if $f(x) = 0$ then the view of $\sA$ is exactly the same as in $H_i$, and otherwise the view of $\sA$ is exactly the same as in $H_{i-1}$. That is, if $f(x)=0$ then $\sA$ outputs $0$ with probability $\frac{1}{q}\sum_{i=1}^{q}p_i$. From the hypothesis that \bbf is unlearnable it follows that $\frac{1}{q}\sum_{i=0}^{q-1}p_i = \frac{1}{2} + \negl$. Similarly, if $f(x) = 1$ then $\sA$ outputs $0$ with probability $\frac{1}{q}\sum_{i=1}^{q}p_{i} = \frac{1}{2} + \negl$. It follows that
\begin{equation}
    |p_0 - p_q| = \left|\sum_{i=0}^{q-1}p_i - \sum_{i=1}^{q}p_{i}\right| = \negl \qedhere
\end{equation}
\end{proof}

Combining \cref{lem:aota} and \cref{thm:qul_qful} gives us the desired result:

\begin{proposition}\label{thm:oracle_inst}
    Let \bbf be an unlearnable balanced binary function (see \cref{def:ulbbf}), then:
    \begin{itemize}
        \item there exists a $\sfdqcp^{1,1}$ secure copy-protection scheme for \bbf relative to a classical oracle, and
        \item there exists a \sfdqcp secure copy-protection scheme for \bbf relative to a quantum oracle.
    \end{itemize}
\end{proposition}

\section{Decoupling}\label{app:decouple}

As we discuss in \cref{sssec:obindccaa}, a desirable property of uncloneable decryptor schemes is that they are \emph{decoupled}. That is, the encryptions of two distinct ciphers are independent (or at least computationally indistinguishable from independent). In particular, we use this property to sidestep the problem of recording a quantum query, which is required for simulating \textsf{CCA2} adversaries while only having access to an encryption oracle.

In this appendix, we show a generic transformation that decouples uncloneable \emph{bit} decryptors and prove that this transformation preserves \sudqx security. For completeness, we also suggest a transformation for decoupling general uncloneable decryptor schemes and sketch a security proof, even though such a transformation is not required for our constructions.

In the context of bit encryption it is tempting to exploit the fact that there are only two possible plaintexts to sidestep this issue by handling an encryption query on input $\tau$ the following way:
\begin{itemize}
    \item query the encryption oracle on $b=0,1$ to obtain ciphertexts $c_b$ and record them, and
    \item apply the unitary $\ket{b,y} \mapsto \ket{b,y\oplus c_b}$ to $\rho$.
\end{itemize}

This allows us to keep a list of all ciphertexts we have generated (this still leaves open the problem of handling decryption calls to ciphertexts which were \emph{not} generated by encryption queries, which we treat in \cref{sssec:ind_cca2}).

The problem with this approach is that it generates a different distribution than directly querying the oracle in superposition. Recall that the randomness of the encryption oracle is sampled \emph{once per query}. Consider an encryption query on a pure state of the form $\ket{\psi} = (\alpha\ket{0}+\beta\ket{1})\otimes \ket{0^\ell}$, then an actual oracle query distributes like $$\comma{\alpha\ket{0,\enc_\sk(0;r)} + \beta\ket{0,\enc_\sk(1;r)}}{r\gets\zo^t}$$ whereas the output of our simulated oracle distributes like $$\comma{\alpha\ket{0,\enc_\sk(0;r_0)} + \beta\ket{0,\enc_\sk(1;r_1)}}{r_0,r_1\gets\zo^t}\text{,}$$ and these distributions might be very different.

We overcome this by \emph{decoupling} the randomness used for each cipher. That is, we make sure that the bits used to generate $c_0$ are disjoint from the bits used to generate $c_1$. This implies that encrypting $0$ and $1$ using the same randomness distributes identically to encrypting them using fresh randomness. 

\begin{definition}\label{def:decoupled}
    An uncloneable bit decryptors scheme \dcsob is called \emph{decoupled} if the encryption procedure is of the form $\enc_\sk(\cdot;r_0,r_1)$ where the output of $\enc_\sk(b)$ depends only on $r_b$.
\end{definition}

Fortunately, it is very simple to generically transform a \sobudqx secure scheme to a decoupled scheme with the same security.

\begin{definition}
    Let $\dcsob$ be an uncloneable bit decryptors scheme. Define the \emph{decoupling} of $\dcsob$ to be the scheme $\dcsobdec$: 
    \begin{itemize}
        \item All procedures but $\enc_\sk$ are defined identically for \dcsob and \dcsobdec.
        \item $\dcsobdec.\enc_\sk$: samples $r_0,r_1\gets \zo^k$ where $t$ is the amount of random bits required by $\dcsob.\enc$, outputs $\dcsob.\enc_\sk(b;r_b)$.
    \end{itemize}
\end{definition}

\begin{lemma}[Decoupling does not Affect Security]\label{thm:decoupling}
    Let \dcsob be an uncloneable bit decryptors scheme. If \dcsob is \sobudqxnk secure (see \cref{def:dcs-obindx}) then \dcsobdec (see \cref{def:decoupled}) is also \sobudqxnk secure.
\end{lemma}

\begin{proof}
    Let $\sA$ be an adversary to the $\sobudqxnk_{\dcsobdec}$ game, we construct an adversary $\sA'$ to the $\sobudqxnk_{\dcsob}$ game such that $$\E\left[\sobudqxnk_{\dcsobdec}(\sA,\lambda)\right] = \E\left[\sobudqxnk_{\dcsob}(\sA',\lambda)\right]\text{.}$$
    
    The adversary $\sA'$ simulates $\sA$ using her own oracle access. They respond to decryption queries by using their own decryption oracle. When getting an encryption query on some input $\rho$, $\sA$ queries her own encryption oracle on both $0$ and $1$ to obtain $c_0$ and $c_1$ respectively. She applies the unitary $\ket{b,x}\mapsto\ket{b,x\oplus c_b}$ to $\rho$ and replies to the query with the result. The distinguishers of $\sA'$ output the output of the distinguishers of $\sA$.
    
    This affords a perfect simulation of the view of $\sA$ in the actual \sobudqxnk game, whereby the output of $\sA'$ distributes identically to the output of $\sA$.
\end{proof}

It is possible to decouple schemes with arbitrary length messages. However, doing so does not solve the problem of recording queries since, in general, the amount of possible queries grows exponentially with the length of the query. For completeness, we informally describe how this could be achieved if we assume that the number of random bits used by the encoding procedure does not grow with the message length (which is not true for our constructions). These ideas could be extended to polynomially many random bits by standard techniques.
    
In \cref{def:decoupled} we required that the ciphers for two different messages would be completely independent even conditioned on the fact that they were generated using the same randomness (in fact, our requirement is strictly stronger). In the context of arbitrary message length, we only require them to be computationally indistinguishable. That is, it is infeasible for a QPT adversary to distinguish an encryption oracle that reuses the same randomness for all queries from an encryption oracle that samples fresh randomness for queries of new plaintexts (but keeps a record of already queried plaintexts to respond consistently).

Assume for simplicity that if $\sk\gets\dcs\keygen(\secparam)$ then $\dcs.\enc_\sk$ requires $\secpar$ random bits. Let $\prf$ be a post-quantum pseudo-random function with the property that if $k\gets \prf.\keygen(\secparam)$ then $\prf.\eval_k$ implements a function from $\zo^*$ to $\zo^\secpar$.
    
Modify $\dcs.\keygen$ such that in addition to $\sk$ it also samples and outputs $k\gets \prf.\keygen(\secparam)$. Replace $\dcs.\enc$ with the circuit $\enc'_{(\sk,k)}(m;r)$ which outputs $c\gets \dcs.\enc_\sk(m;\prf.\eval(m,r))$.
    
The security property of $\prf$ then implies $\prf.\eval(m,r)$ is indistinguishable from using a truly random string. In particular, it is computationally impossible to distinguish the unitary $\ket{m,x}\mapsto \ket{m,x\oplus \enc'_{(\sk,k)}(m)}$ from an oracle which on each query chooses a random string $r_m$ for any possible message $m$ an applies the unitary $\ket{m,x}\mapsto \ket{m,x\oplus \enc_{\sk}(m;r_m)}$.

\section{Weakly Copy-Protectable \textsf{BBF}s are Weak \textsf{PRF}s}\label{app:ggm}

In this appendix, we provide a full proof to \cref{cor:prf}. Namely, we show that if it is possible to construct a $\swqcp^{1,1}$ secure copy-protection scheme for \bbf, then \bbf must be a weak $\prf$.

To do so, we first recall the definition of a weak $\prf$ for binary functions. For any function $f:D \to R$, let $\sR(f)$ denote an oracle which takes no input and outputs a pair of the form $(r,f(r))$ where $r\gets D$.

\begin{definition}\label{def:wprf}
    Let \bbf be a binary function. We say that \bbf is a post-quantum \emph{weak pseudo-random function} (or \emph{weak \prf}) if for any QPT oracle machine $\sA$ whose output is one bit it holds that
    $$\prob{\sA^{\sR(\sO_b)}=b} = \frac{1}{2} + \negl$$
    where $\sk \gets \bbf.\samp(\secpar)$, $\sO_0 = \bbf.\eval_\sk$, $\sO_1$ is a uniformly random function with the same input and output lengths as $\bbf.\eval_\sk$ and $b\gets \zo$.
\end{definition}

The choice of the notation \bbf in \cref{def:bbf} is suggestive of the fact that any weak $\prf$ must be balanced. Indeed, suppose the underlying function outputs $b$ with probability significantly greater than $\frac{1}{2}$. In that case, an adversary which makes a single query and outputs $0$ iff the output is $b$ is a counterexample to \cref{def:bbf}.

\begin{remark}
    The key difference between weak and standard $\prf$s is that the former only allows evaluating the function on random inputs, hence the random input oracles which appear in \cref{def:wprf}. If instead we have granted $\sA$ oracle access $\sO_0$ and $\sO_1$ we would have recovered the definition of a \emph{post-quantum $\prf$}, and if we have granted access to $\ket{\sO_0}$ and $\ket{\sO_1}$ we would have recovered the definition of \emph{fully quantum $\prf$}. For a discussion of these two notions and the differences thereof (including an explicit example of a post-quantum \prf which is not fully quantum) see \cite{Zha12}.
\end{remark}

The second ingredient we require is notion of security for symmetric encryption called \emph{left-or-right} (\textsf{LoR}) security, first introduced in \cite{BDJR97}. Under this notion, the challenger $\sC$ first samples a random coin $b$. The adversary may make several challenge queries of the form $(m_0,m_1)$ whose response is the encryption of $m_b$, and their goal is to guess $b$. We emphasize that $b$ is chosen once for the entire game.

\begin{definition}[\slorcpa security]\label{def:lorcpa}
    Let $\se$ be a symmetric encryption scheme, and let $\sA$ be a procedure, the \emph{$\slorcpa_\se(\sA,\secpar)$} game is defined as follows:
    \begin{itemize}
        \item $\sC$ samples $\sk\gets \se.\keygen(\secparam)$, and $b\gets \zo$,
        \item $\sA$ transmits to $\sC$ pairs of the form $(m_0,m_1)$ with $|m_0|=|m_1|$,
        \item $\sC$ responds to each pair with $\enc_\sk(m_0)$,
        \item $\sA$ outputs a bit $\beta$,
        \item $\sA$ wins the game if $b=\beta$.
    \end{itemize}
    The scheme \se is considered post-quantum \emph{\slorcpa secure} if for any QPT $\sA$, the winning probability is at most $1/2 + \negl$.
\end{definition}

It might come off as curious that the adversary for the \slorcpa game enjoys no oracle access despite the \textsf{CPA} in the name indicating this game should simulate a chosen plaintext attack. However, providing $\sA$ with access to $\enc_\sk$ does not increase their strength in any way, as she can simulate a query on plaintext $m$ by sending the pair $(m,m)$ to $\sC$.

It might seem that \slorcpa security is stronger than \sindcpa security, as the adversary is more powerful. However, \cite[Theorem~4]{BDJR97} shows that these added capabilities can only increase the advantage by a polynomial multiplicative factor, so that \sindcpa security does imply \slorcpa.

We relate balanced binary functions to symmetric encryption through the following construction:

\begin{construction}\label{con:sebbf}
    For any $\bbf$ with input length $\ell$ define the bit encryption scheme $\se_\bbf$:
    \begin{itemize}
        \item $\se_\bbf.\keygen \equiv \bbf.\samp$,
        \item $\se_\bbf.\enc_\sk(b)$ outputs $(r,b\oplus \bbf.\eval_\sk(r))$ where $r\gets\zo^\ell$, and
        \item $\se_\bbf.\dec_\sk((r,\beta))$ outputs $\beta\oplus \bbf.\eval_\sk(r)$.
    \end{itemize}
\end{construction}

The correctness of $\se_\bbf$ follows from the correctness of \bbf. Note that it holds that $\se_\bbf = \se_{\dcsobcca}$ (recall the definition of the underlying scheme \cref{def:assoc-se}) where \dcsobcca is the uncloneable decryptor scheme of \cref{con:dcs_cca1_res} instantiated with \bbf. This gives us the following series of implications:

\begin{itemize}
    \item say that there exists a copy-protection scheme \sqcp for \bbf which is $\swqcp^{1,1}$ secure,
    \item it follows from \cref{thm:con_1bit_dcs_cca1} that using \sqcp to instantiate \cref{con:dcs_cca1_res} is $\sobudqcpa^{1,1}$ secure,
    \item it follows from \cref{thm:DCS-IND} that the underlying scheme to this construction is \sindqcpa secure, and in particular \sindcpa secure,
    \item from the discussion above, the underlying scheme is also \slorcpa secure, and
    \item also from the discussion above, this underlying scheme is exactly $\se_\bbf$.
\end{itemize}

Hence, if $\bbf$ could be copy-protected with $\swqcp^{1,1}$ security then $\se_\bbf$ is post-quantum \slorcpa secure. So in order to prove \cref{cor:prf} it suffices to prove:

\begin{lemma}\label{lemma:bbfwprf}
    If $\se_\bbf$ is post-quantum \slorcpa secure then \bbf is a post-quantum weak \prf.
\end{lemma}

\begin{remark}
    The converse also holds. The analysis of \cite[Proposition~5.3.19]{Gol04} shows that $\se_\bbf$ is classically \sindcpa secure whenever \bbf is a \prf. By examining the proof one notes that:
    \begin{enumerate}
        \item the reduction therein only queries \bbf on uniformly random points, whereby the proof holds verbatim for weak $\prf$s as well, and
        \item if we assume that \bbf is a post-quantum weak $\prf$, then the same argument without modification shows that $\se_\bbf$ is post-quantum \sindcpa secure.
    \end{enumerate}
\end{remark}

\begin{proof}
    Let $\sA$ be a QPT oracle machine such that
    $$\prob{\sA^{\sR(\sO_b)}=b} = \frac{1}{2} + \epsilon\text{,}$$ where the oracles $\sO_0,\sO_1$ are as described in \cref{def:wprf}. We show that $\epsilon = \negl$ by constructing an adversary for the $\slorcpa_{\se_\bbf}$ game with advantage $\epsilon - \negl$. It will then follow from the \slorcpa security of $\se_\bbf$ that $\epsilon = \negl$.
    
    The adversary $\sA'$:
    \begin{itemize}
        \item simulates $\sA$,
        \item whenever $\sA$ makes a query, $\sA'$ sends the pair $(0,z)$ with $z\gets\zo$ to $\sC$ and replies to $\sA$ with the respond she got from $\sC$,
        \item $\sA'$ resumes the simulation until $\sA$ outputs a bit $\beta$, and outputs $\beta$.
    \end{itemize}
    
    Let $b$ be the bit sampled by $\sC$ at the beginning of the game, then $\sA'$ wins iff $b=\beta$.
    
    Note that if $b=0$ then the queries of $\sA$ are always answered with pairs of the form $(r,\bbf_\sk(r)\oplus 0) = (r,\bbf_\sk(r))$, so her view distributes exactly as if she had access to the oracle $\sO_0$.
    
    If $b=1$ then the queries are always answered with pairs of the form $(r,\bbf_\sk(r)\oplus z)$ where $z\gets \zo$. Hence, the second bit distributes uniformly. The view of $\sA$ in this scenario is identical to her view given access to the $\sO_1$ oracle conditioned on the event that $\sA'$ never responded to two different queries with the same $r$, which holds with overwhelming probability.
    
    It follows that the view of $\sA$ in the simulation above is negligibly close to her view given access to the actual oracles $\sO_0,\sO_1$.
    
    It follows that the probability that $\sA$ wins is negligibly close to $\prob{\sA^{\sR(\sO_b)}=b} = \frac{1}{2} + \epsilon$. That is, the advantage of $\sA$ is $\epsilon - \negl$.
\end{proof}

\section{Fully Specified Constructions}\label{app:fleshed}

We presented our construction for \sudqcpa and \sudqccaa as transformations of given schemes with properties obtained by previous constructions. This appendix presents full specifications of these constructions instantiated from a \sfdqcp secure copy-protection and a deterministic \sseufcma secure digital signature.

\subsection{\textsf{UD-qCPA} secure construction}

Let $\sqcp_\bbf$ be a quantum copy-protection scheme for a binary balanced function \bbf with input length $\ell_\bbf$. We define an uncloneable decryptors scheme \dcscpa:
\begin{itemize}
    \item $\dcscpa.\keygen \equiv \bbf.\samp$.
    \item $\dcscpa.\decgen \equiv \sqcp_\bbf.\ptect$.
    \item $\dcscpa.\enc_\sk(m)$ outputs $(c_1,\ldots,c_{|m|})$ where 
    \[\begin{aligned}
        r_i &\gets \zo^{\ell_\bbf}\\
        c_i & = (r_i,m_i\oplus \bbf.\eval_\sk(r_i))
    \end{aligned}\text{.}\]
    \item $\dcscpa.\dec_\sk((r_1,\beta_1),\ldots,(r_\ell,\beta_\ell))$ outputs
    $$\beta_1\oplus \bbf.\eval_\sk(r_1) \| \ldots \| \beta_\ell\oplus\bbf.\eval_\sk(r_\ell)\text{.}$$
    \item $\dcscpa.\qdec(\rho,(r_1,\beta_1),\ldots,(r_\ell,\beta_\ell))$ outputs
    $$\beta_1\oplus \sqcp_\bbf.\qeval(\rho,r_1) \| \ldots \| \beta_\ell\oplus\sqcp_\bbf.\qeval(\rho,r_\ell)\text{.}$$
\end{itemize}

\cref{thm:dcs_cca1} implies that if $\qcp_\bbf$ is \sfdqcpnk (resp. \sfdqcp) secure then \dcscpa is \sudqcpank (resp. \sudqcpa) secure.

\subsection{\textsf{UD-qCCA2} secure construction}

Let $\sqcp_\bbf$ be a quantum copy-protection scheme for a binary balanced function \bbf with input length $\ell_\bbf$. Let \ds be a deterministic digital signature scheme.

We define an uncloneable decryptors scheme \dcsccaa:
\begin{itemize}
    \item $\dcsccaa.\keygen(\secparam)$ outputs $(\sk_\bbf,\sk_\ds,\pk_\ds)$ where
    \[\begin{aligned}
        \sk_\bbf & \gets \bbf.\samp(\secparam)\\
        (\sk_\ds,\pk_\ds) &\gets \ds.\keygen(\secparam)
    \end{aligned}\text{.}\]
    \item $\dcsccaa.\decgen(\sk)$ parses $\sk$ as $(\sk_\bbf,\sk_\ds,\pk_\ds)$, outputs $\rho\otimes \pk_\ds$ where $\rho\gets \sqcp_\bbf.\ptect(\sk_\bbf)$.
    \item $\dcsccaa.\enc_\sk(m)$ parses $\sk$ as $(\sk_\bbf,\sk_\ds,\pk_\ds)$, outputs $$(r,(c_1,s_1),\ldots,(c_{|m|},s_{|m|})$$ where 
    \[\begin{aligned}
        r & \gets \zo^\secpar \\
        c_i & = (r_i,m_i\oplus \bbf.\eval_\sk(r_i))\\
        s_i & = \ds.\sign_{\sk_\ds}(c_i,|m|,i,r)\\
        r_i &\gets \zo^{\ell_\bbf}
    \end{aligned}\]
    \item $\dcsccaa.\dec_\sk((r,(r_1,\beta_1),\ldots,(r_\ell,\beta_\ell)))$ parses $\sk$ as $(\sk_\bbf,\sk_\ds,\pk_\ds)$, outputs
    $$\begin{cases}
                b_1\|\ldots\|b_\ell & \comma{\forall i=1,\ldots,\ell}{\ds.\ver_{\pk_\ds}((c_i,\ell,i,r),s_i) = 1}\\
                \bot & \mbox{else} \\
            \end{cases}$$ where $b_i = \beta_i \oplus \bbf.\eval{\sk_\bbf}(r_i)$.
    \item $\dcsccaa.\qdec(\rho,(r,(r_1,\beta_1),\ldots,(r_\ell,\beta_\ell)))$ outputs
            $$\begin{cases}
                b_1\|\ldots\|b_\ell & \comma{\forall i=1,\ldots,\ell}{\ds.\ver_{\pk_\ds}((c_i,\ell,i,r),s_i) = 1}\\
                \bot & \mbox{else} \\
            \end{cases}$$ where $b_i = \beta_i \oplus \sqcp_\bbf.\qeval(\rho,r_i)$.
\end{itemize}

\cref{thm:udccaank} implies that if $\qcp_\bbf$ is \sfdqcpnk (resp. \sfdqcp) secure and \ds is \sseufcma secure then \dcscpa is \sudqccaank (resp. \sudqccaa) secure.

\ifnum\masterthesis=1
    \bibliography{main}
\fi

\end{document}